\documentclass[a4paper,11pt]{article}
\usepackage[utf8]{inputenc}

\usepackage{geometry}
\usepackage{xcolor}
\usepackage{amsmath, array, amssymb, amsfonts,amsthm}
\usepackage{xfrac}
\usepackage[inline]{enumitem}

\usepackage[french, english]{babel}
\usepackage[all]{xy}

\usepackage{txfonts}  

\usepackage{sectsty}
\sectionfont{\centering}
\subsectionfont{\centering}
\subsubsectionfont{\centering}

\usepackage{booktabs}
\usepackage{caption}
\usepackage{dsfont}
\usepackage{mathtools}
\usepackage[hidelinks]{hyperref}

\usepackage{textcomp}
\usepackage{multicol}
\usepackage[square,numbers,sort,compress,semicolon,merge]{natbib}
\let\cite\citep 
\usepackage{hyperref}
\hypersetup{colorlinks=true, urlcolor=blue, citecolor=blue, linktoc=page}

\makeatletter
\renewcommand*\env@matrix[1][\arraystretch]{%
  \edef\arraystretch{#1}%
  \hskip -\arraycolsep
  \let\@ifnextchar\new@ifnextchar
  \array{*\c@MaxMatrixCols c}}
\makeatother

\renewcommand\P{\mathcal{P}}
\newcommand\M{\mathcal{M}}

\newcommand\doubleR{\mathbb{R}}
\newcommand\doubleC{\mathbb{C}}
\renewcommand\1{\textbf{1}}
\newcommand\id{\textit{id}}
\newcommand\T{\mathcal{T}}

\renewcommand\H{\mathcal{H}}

\newcommand\E{\mathcal{E}}

\renewcommand\S{\mathcal{S}}

\newcommand\U{\mathcal{U}}
\newcommand\SO{\mathcal{SO}}

\newcommand\K{\mathcal{K}}
\newcommand\J{\mathcal{J}}

\newcommand\W{\mathcal{W}}
\newcommand\vphi{\varphi}

\renewcommand\u{u_{\text{\tiny L}}}
\renewcommand\l{\text{\tiny{L}}}
\newcommand\ww{\text{\tiny{W}}}
\renewcommand\ss{\text{s}}
\newcommand\n{\text{\tiny{N}}}
\newcommand\sP{\mathsf{P}}
\newcommand\sC{\mathsf{C}}
\newcommand\sT{\mathsf{T}}
\newcommand\sW{\mathsf{W}}
\newcommand\sS{\mathsf{S}}
\newcommand\sE{\mathsf{E}}

\renewcommand\epsilon{\varepsilon}

\newcommand\rarrow{\rightarrow}

\newcommand\LieG{\mathfrak{g}}
\newcommand\LieH{\mathfrak{h}}

\newcommand\co{\mathfrak{co}}

\newcommand\LieK{\mathfrak{k}}
\newcommand\LieJ{\mathfrak{j}}

\renewcommand\t{\widetilde}
\newcommand\h{\widehat}
\renewcommand\b{\bar }
\newcommand\w{\wedge}
\renewcommand\d{\partial}
\newcommand\s{\sigma}

\renewcommand\-{^{-1}}
\newcommand\Ad{\text{Ad}}
\newcommand\ad{\text{ad}}
\renewcommand\id{\text{id}}
\renewcommand\1{\mathds{1}}

\DeclareMathOperator{\Aut}{Aut}
\DeclareMathOperator{\Tr}{Tr}

\newtheorem{thm}{Theorem}

\newtheorem{cor}[thm]{Corollary}
\newtheorem{prop}[thm]{Proposition}
\theoremstyle{definition}

\hypersetup{
pdfauthor={J. Attard, J. François},
pdftitle={Tractors and local Twistors from conformal Cartan geometry: a gauge theoretic approach},
pdfsubject={},
pdfcreator={pdflatex},
pdfproducer={pdflatex},
pdfkeywords={}
}

\begin{document}

\title{Tractors and Twistors from conformal Cartan geometry: a gauge theoretic approach \\ 
\center I. Tractors}
\author{J. Attard,${\,}^a$ \hskip 0.7mm J. François${\,}^b$ }
\date{}

\maketitle
\begin{center}

\vskip -0.5cm
${}^a$ Aix Marseille Univ, Université de Toulon, CNRS, CPT, Marseille, France\\
${}^b$ Corresponding author: jordanfrancois54@gmail.com
\end{center}

\begin{abstract}
Tractors and Twistors bundles both provide natural conformally covariant calculi on $4D$-Riemannian manifolds. They have different origins but are closely related, and usually constructed bottom-up from prolongation of defining differential equations. We propose alternative top-down gauge theoretic constructions starting from the conformal Cartan bundle $\P$ and its vectorial $E$ and spinorial $\sE$ associated bundles.
Our key ingredient is the dressing field method of gauge symmetry reduction, which allows to exhibit tractors and twistors and their associated connections as gauge fields of a non-standard kind as far as Weyl rescaling symmetry is concerned. By which we mean that they implement the gauge principle but are of a different geometric nature than the well known differential geometric objects usually underlying gauge theories.   
We provide the corresponding BRST treatment.  The present paper deals with the case of tractors while  a companion paper deals with twistors. 
\end{abstract}


\vspace{1mm}


\vspace{1.5mm}


\tableofcontents

\section{Introduction}  

Following Einstein's General Relativity (GR), the interest of local scale symmetry for physics began with Weyl's $1918$ unified theory of gravity and electromagnetism \cite{Weyl1918, Weyl1919, ORaif} which is also the root of the very idea of gauge symmetry. In $1921$ Bach proposed  what is now called  Weyl or conformal gravity \cite{Bach1921}. 
 Shortly after the inception of Yang-Mills' $1954$ non-abelian theory, Utiyama pioneered the  systematic gauging of global symmetry groups \cite{Utiyama1956, ORaif} and opened the way for a gauge theoretic formulation of gravitation.
  Some twenty years later, in the wake of the interests in supersymmetry, several authors investigated formulations of gravity and supergravity, gauging Poincaré and deSitter groups and their supersymmetric extensions \cite{Cham-West1977,Townsend1977,McDowell-Mansouri,West1978,Stelle-West1979}. Some others addressed the gauging of the full $15$-parameters conformal group $SO(2, 4)$,\footnote{First introduced to physics, according to \cite{Fulton_et_al1962}, by \cite{Cunningham} and \cite{Bateman1909, Bateman1910} in connection with Special Relativity and the invariance of Maxwell's equations.}
   which extends the $10$-parameters Poincaré group by $4$ special conformal transformations - also known as inversions or conformal boosts - and Weyl rescalings  \cite{CrispimRomao1977, Ferber-Freund1977, Kaku_et_al1977, Kaku-et-al1978, Harnad-Pettitt1, Harnad-Pettitt2}.  Since the mid $80$'s to this day conformal symmetry is central to subjects generating a vast literature such as phase transitions in statistical physics, conformal field theory (CFT) and string theory. 

In the midst of this story, in the $60$'s and early $70$'s, Penrose followed its own original path and devised twistor theory as an alternative framework for physics, and quantum gravity, in which conformal symmetry is pivotal \cite{Penrose1960, Penrose1967, Penrose1968, Penrose-McCallum72}. In twistor theory spinors quantities takes over the role of tensors and the - conformally compactified - Minkowski space is seen as secondary, emerging from a more fundamental twistor space hoped to be more fit for quantization.  A generalization to arbitrary pseudo-riemannian manifolds gave rise to the concept of local twistors, which provides a conformal spinorial calculus. 
Twistor theory remains an active area of research in physics. 
\smallskip

Parallel progresses unfolded in mathematics. The renewal of differential geometry  through the development of the theory of connections sparked by Einstein's GR and the current of ideas aiming at its improvement by unified field theories, is well documented (see e.g \cite{Goenner2004}). It started with the theory of parallel transport of Levi-Civita and Schouten in $1917$ and $1918$ and came to a climax in $1922$-$1926$ with Cartan's notion of moving frames and $\emph{espace généralisés}$. Those are manifolds with torsion in addition to curvature and classic examples are manifolds endowed with projective and conformal connections. The next important step was achieved in the late $30$'s when Whitney gave the first definition of fibered manifolds. Then, in the late $40$'s and early $50$'s Ehresman, student of Cartan, proposed a notion of connection on fibered manifolds, superseding that of Cartan connection by its abstractness and generality. During the following twenty years, the geometry of Ehresmann connections on fiber bundles came to be recognized as the mathematical underpinning of Yang-Mills theories. After being largely overlooked, Cartan connections were more recently acknowledged as the adequate framework for gravitational theories, see e.g \cite{Wise09, Wise10}.

 Early in this story, in $1925$-$26$, Thomas independently developed a calculus on conformal (and projective) manifolds, analogous to Ricci calculus on Riemannian manifolds, alternative to Cartan's viewpoint. His work was rediscovered and expanded in \cite{Bailey-et-al94} where it was given its modern guise as a vector bundle  called \emph{standard tractor bundle} endowed with a linear connection, the \emph{tractor connection}. In recent years this conformal tractor calculus has been of interest for physicists, see e.g \cite{Gover-Shaukat-Waldron09, Gover-Shaukat-Waldron09-2}
\medskip

Tractors and twistors are closely linked. It has been noticed that both define vector bundles associated to the conformal Cartan principal bundle $\P(\M, H)$, with $H$ the parabolic subgroup of $SO(2, 4)$\footnote{In the case of tractors this generalizes to $SO(r+1, s+1)$ if $\M$ is a Lorentzian $(r, s)$-manifold. For twistors $\M$ must be $4$-dimensional.} comprising Lorentz, Weyl and conformal boost symmetries, and that tractor and twistor connections are induced by the so-called \emph{normal Cartan connection} $\varpi_\n$ on $\P$. See e.g \cite{Friedrich77} for the twistor case. 
 
In standard presentations however, both tractors and local twistors are constructed through the prolongation of defining differential equations defined on a Riemannian manifold $(\M, g)$: the Almost Einstein (AE) equation and twistor equation respectively. The systems thus obtained are linear and closed, so that they can be rewritten as  linear operators acting on multiplets of variables called parallel tractors and global twistors respectively. The  behavior of the latter under Weyl rescaling of the underlying metric is given by definition and commutes with the actions of their associated  linear operators, which are then respectively called tractor and twistor connections. The multiplets are then seen as parallel sections  of  vectors bundles, the tractor and local twistor bundles, endowed with their linear connections. 
For the procedure in the tractor case see \cite{Bailey-et-al94}, or \cite{Curry-Gover2015} for a more recent and detailed review. For the twistor case see the classic \cite{Penrose-Rindler-vol2}, or even \cite{Bailey-Eastwood91} which generalizes the twistor construction to paraconformal (PCF) manifolds. 

This constructive procedure via prolongation has been deemed more explicit \cite{Bailey-Eastwood91}, more intuitive and direct \cite{Bailey-et-al94} than the viewpoint in terms of vector bundles associated to $\P(\M, H)$. Since it starts from $(\M, g)$ to built  a gauge structure on top of it - vector bundles endowed with connections - we may call-it a ``bottom-up'' approach.  
\medskip

Here and in a companion paper, we would like to put forward a ``top-down'' approach to tractors and twistors that relies on a gauge theoretic method of gauge symmetry reduction: the dressing field method. Given a gauge structure (fiber bundles with connections) on $\M$, this method allows to systematically construct partially gauge-invariant composite fields built from the usual gauge fields and a so-called \emph{dressing field}. According to the transformations of the latter under the residual gauge symmetry, the composite fields  display interesting properties. In a noticeable case they are actually gauge fields of non-standard kind, meaning that they implement the gauge principle but are not of the same geometric nature as  the usual Yang-Mills fields. The dressing field method fits in the BRST framework. 

When the gauge structure on $\M$ is the conformal Cartan bundle $\P(\M, H)$ endowed with a Cartan connection and two vector bundles $E$ and $\sE$ associated to the defining representation $\doubleR^6$ and spin representation $\doubleC^4$ of $SO(2, 4)$, the dressing field method allows to erase the conformal boost gauge symmetry. The composite fields obtained by dressing of the sections of $E$ and $\sE$ are exactly tractors and twistors, while the dressed Cartan connection straightforwardly induces generalized tractor and twistor covariant derivatives. If the normal Cartan connection is dressed, these reduce exactly to the tractor and twistor derivatives. We stress that tractors and twistors thus obtained, while being genuine standard gauge fields with respect to (w.r.t) Lorentz gauge symmetry, are examples of non-standard gauge field alluded to above w.r.t Weyl gauge symmetry. This, we think, is a new consideration worth emphasizing. 
\medskip

Twistors will be dealt with in  a companion paper. The present one focuses on tractors and is organized as follows. In section \ref{The geometry of gauge fields} we review the basics of differential geometry underlying gauge theories, including Cartan geometry, as well as the BRST formalism, so as to fix notations and define useful concepts. 

In section \ref{Reduction of gauge symmetries: the dressing field method} we review in great details the dressing field method of gauge symmetry reduction. We prove a number of general propositions for subsequent use, with special emphasis on the emergence of the composite fields as gauge fields of a non-standard kind. We also cover the local aspects, necessary for physics, and provide the BRST treatment. This is the most comprehensive presentation of the method yet, besides  \cite{Francois2014}.

Then in section \ref{Tractors from conformal Cartan geometry via dressing} we put this material to use: after a brief review of the ``bottom-up'' procedure for tractors, we construct tractors and tractor connection ``top-down'' from the conformal Cartan bundle through two successive dressing operations. Residual Lorentz and Weyl gauge symmetries are analyzed both at the finite and BRST level. 
We summarize our main results and gather our comments  in our conclusion \ref{Conclusion}.

\section{The geometry of gauge fields}   
\label{The geometry of gauge fields}   

Gauge theories are a cornerstone of modern physics built on the principle that the fundamental interactions originate from local symmetries called gauge symmetries. 
The mathematics underlying classical gauge theories is now widely known to be the differential geometry of fiber bundles and connections supplemented by the differential algebraic BRST approach.
Of these we briefly recall the basic features in this section, if only to fix notations, before exposing the dressing field method in the next.

\subsection{Basic differential geometry} 
\label{Basic differential geometry} 

Let $\P(\M, H)$ be a principal fiber bundle over a smooth $n$-dimensional manifold $\M$, with structure Lie group $H$ and projection map $\pi: \P \rarrow \M$.  Given a representation $(\rho, V)$ for $H$ we have the associated bundle $E:=P\times_\rho V$, whose sections are in bijective correspondence with $\rho$-equivariant maps on $\P$: $\t\vphi \in \Gamma(E) \leftrightarrow \vphi \in \Lambda^0(\P, \rho)$.

Given the right-action  $R_hp=ph$ of $H$ on $\P$, a  $V$-valued $n$-form $\beta$ is  said $\rho$-equivariant iff $R^*_h\beta=\rho(h\-)\beta$. Let $X^v\in V\P\subset T\P$ be a vertical vector field induced by the infinitesimal action of $X \in \LieH=$ Lie$H$ on $\P$. A form $\beta$ is said horizontal if $\beta(X^v, \ldots)=0$. A form $\beta$ is said $(\rho, V)$-tensorial if it is both horizontal and $\rho$-equivariant. 
\smallskip

Let $\omega \in \Lambda^1(\P, \LieH)$ be a choice of connection on $\P$: it is $\Ad$-equivariant and satisfies $\omega(X^v)=X$. The horizontal subbundle $H\P \subset T\P$, the non-canonical complement of $V\P$, is defined by $\ker \omega$. Given $Y^h \in H\P$ the horizontal projection of a vector field $Y \in T\P$, the covariant derivative of a $p$-form $\alpha$ is defined by $D\alpha:=d\alpha(Y_1^h, \ldots, Y_p^h)$. 

The connection's curvature form $\Omega \in \Lambda^2(\P, \LieH)$  is defined as its covariant derivative, but is algebraically given by the Cartan structure equation $\Omega=d\omega+\tfrac{1}{2}[\omega, \omega]$. 
Given a $(\rho, V)$-tensorial $p$-forms $\beta$ on $\P$,  its covariant derivative  is a $(\rho, V)$-tensorial $(p+1)$-form  algebraically given by $D\beta= d\beta +\rho_*(\omega)\beta$.  Furthermore  $D^2\beta=\rho_*(\Omega)\beta$. 

Since $\vphi$ is a $(\rho, V)$-tensorial $0$-form, its covariant derivative is the $(\rho, V)$-tensorial $1$-form $D\vphi:=d\vphi + \rho_*(\omega)\vphi$. The section $\vphi$ is said parallel if $D\vphi=0$. One can show that the curvature $\Omega$ is a $(\Ad, \LieH)$-tensorial $2$-form, so its covariant derivative is $D\Omega=d\Omega +\ad(\omega)\Omega=d\Omega +[\omega, \Omega]$. Given the Cartan structure equation, this vanishes identically and provides the Bianchi identity $D\Omega=0$.

Given a local section $\s: \U \subset \M \rarrow \P$, we have that $\s^*\omega \in \Lambda^1(\U, \LieH)$ is a Yang-Mills gauge potential, $\s^*\Omega \in \Lambda^2(\U, \LieH)$ is the Yang-Mills field strength and $\s^*\vphi$ is a matter field, while $\s^*D\vphi=d\s^*\vphi+\rho_*(\s^*\omega)\s^*\vphi$ is the minimal coupling of the matter field to the gauge potential. 
\smallskip

The group of vertical automorphisms of $\P$, $\Aut_v(\P):=\left\{\Phi:\P \rarrow \P \ |\  h\in H, \Phi(ph)=\Phi(p)h \text{ and } \pi \circ \Phi= \Phi \right\}$ is isomorphic to  the gauge group $\H:=\left\{ \gamma :\P \rarrow H\ | \  R^*_h\gamma(p)=h\- \gamma(p) h  \right\}$, the isomorphism being $\Phi(p)=p\gamma(p)$. The composition law of $\Aut_v(\P)$, $\Phi_2^*\Phi_1:=\Phi_1 \circ \Phi_2$, implies that the gauge group acts on itself by $\gamma_1^{\gamma_2}:=\gamma_2^{-1} \gamma_1 \gamma_2$. 

The gauge group $\H \simeq \Aut_v(\P)$ acts on the connection, curvature and $(\rho, V)$-tensorial forms as,
\begin{align}
\label{ActiveGT}
&\omega^\gamma:=\Phi^*\omega=\gamma\-\omega\gamma + \gamma\- d\gamma, \quad \Omega^\gamma:=\Phi^*\Omega=\gamma\-\Omega \gamma, \\ 
&\vphi^\gamma:= \Phi^*\vphi=\rho(\gamma\-)\vphi, \quad  \text{and} \quad (D\vphi)^\gamma=D^\gamma \vphi^\gamma=\Phi^*D\vphi=\rho(\gamma\-)D\vphi.\notag
\end{align}
These are \emph{active} gauge transformations, formally identical but to be conceptually distinguished from \emph{passive} gauge transformations relating two local descriptions of the same global objects. Given two local sections related via  $\s_2=\s_1 h$, either over the same open set $\U$ of $\M$ or over the overlap of two open sets $\U_1 \cap \U_2$, one finds
\begin{align}
\label{PassiveGT}
&\s_2^*\omega=h\-\s_1^*\omega\  h + h\- dh, \quad \s_2^*\Omega=h\-\s_1^*\Omega\  h, \\ 
&\s_2^*\vphi= \rho(h\-)\s_1^*\vphi, \quad  \text{and} \quad \s_2^*D\vphi=\rho(h\-)\s_1^*D\vphi.\notag
\end{align}
This distinction between active vs passive gauge transformations is reminiscent of the distinction diffeomorphism vs coordinate transformations in General Relativity. 
\smallskip

If the manifold is equipped with a $(r, s)$-Lorentzian metric allowing for a Hodge star operator, and if $V$ is equipped with an inner product $\langle\ , \rangle$, then the prototypical Yang-Mills Lagrangian $m$-form  for a gauge theory is 
\begin{align*}
L(\s^*\omega, \s^*\vphi)=\tfrac{1}{2}\Tr[\s^*\Omega \w * (\s^*\Omega) ]+ \langle \s^*D\vphi, *\s^*D\vphi\rangle - U(\s^*\vphi),
\end{align*}
where $U$ is a potential term for the matter field, as is necessary for the spontaneous symmetry breaking (SSB) mechanism in the electroweak sector of the Standard Model. 

\subsection{Cartan geometry} 
\label{Cartan geometry} 

Connections $\omega$ on $\P$ such as described, known as Ehresmann or principal connections, are well suited to describe Yang-Mills fields of gauge theory. They are the heirs of another notion of connection, best suited  to describe gravity in a gauge theoretical way: Cartan connections. 
A Cartan connection $\varpi$ on a principal bundle $\P(M, H)$, beside satisfying the two defining properties of a principal connection,  defines an absolute parallelism on $\P$. A bundle equipped with a Cartan connection is a Cartan geometry, noted $(\P, \varpi)$. 

Explicitly, given a Lie algebra $\LieG \supset \LieH$ with dim $\LieG=$ dim $T_p\P$ for which a group is not necessarily chosen, a Cartan connection is $\varpi \in \Lambda^1(\P, \LieG)$ satisfying: $\varpi (X^v)=X$, $R^*_h\varpi =\Ad_{h\-}\varpi$ and $\varpi_p : T_p\P \rarrow \LieG$ is a linear isomorphism $\forall p \in \P$. This last defining property implies that the geometry of the bundle $\P$ is much more intimately related to the geometry of the base spacetime manifold $\M$, hence the fitness of Cartan geometry to describe gravity in the spirit of Einstein's insight. Concretely one can show that $T\M \simeq \P\times_H \LieG/\LieH$, and the image of $\varpi$ under the projection $\tau: \LieG \rarrow \LieG/\LieH$ defines a generalized soldering form, $\theta:=\tau(\varpi)$. The latter, more commonly known as the vielbein in the physics literature, implements (a version of) the equivalence principle and accounts for the specificities of gravity among other gauge interactions.  The $(\Ad, \LieG)$-tensorial curvature $2$-form $\b\Omega$ of $\varpi$ is defined through the Cartan structure equation: $\b\Omega=d\varpi + \tfrac{1}{2}[\varpi, \varpi]$. Its $\LieG/\LieH$-part is the torsion $2$-form $\Theta:=\tau(\b\Omega)$.

Given a $\Ad_H$-invariant bilinear form $\eta$ of signature $(r, s)$  on $\LieG/\LieH$, a $(r, s)$-metric $g$  on $\M$ is induced  via  $\varpi$ according to $g(X, Y):=\eta\left( \s^*\theta(X), \s^*\theta(Y)\right)$, for $X, Y \in T\M$ and $\s :\U\subset \M \rarrow \P$ a trivializing section. 

In the case $\LieG$ admits a $\Ad_H$-invariant splitting $\LieH + \LieG/\LieH$, the Cartan geometry is said reductive. Then one has $\varpi=\omega+\theta$, where $\omega$ is a principal $H$-connection, and $\b\Omega=\Omega + \Theta$ with $\Omega$ the curvature of $\omega$. 
As an example, the Cartan geometry with $(\LieG,\LieH)$ the Poincaré and Lorentz Lie algebras is Riemann geometry with torsion.

Given a group $G$ and a closed subgroup $H$, $G/H$ is a homogeneous manifold and $G \xrightarrow{\pi} G/H$ is a $H$-principal bundle. The Maurer-Cartan form $\varpi_\text{\tiny G}$ on $G$ is a flat Cartan connection. So $(G, \varpi_\text{\tiny G})$ is a flat Cartan geometry, sometimes referred to as the Klein model for the geometry $(\P, \varpi)$,  which is thus said to be of type $(G, H)$. 

Let $V$ be a $(\LieG, H)$-module, i.e it supports a $\LieG$-action $\rho_*$ and a $H$-representation $\rho$ whose differential coincides with the restriction of the  $\LieG$-action  to $\LieH$. The Cartan connection defines a covariant derivative on $(\rho, V)$-tensorial forms. On sections of associated bundles, i.e on $\rho$-equivariant maps $\vphi$, we have: $D\vphi:=d\vphi + \rho_*(\varpi)\vphi$. As usual $D^2 \vphi=\rho_*(\b\Omega)\vphi$. On the curvature it gives the Bianchi identity: $D\b\Omega=d\b\Omega+[\varpi, \b\Omega]=0$.

The gauge group $\H \simeq \Aut_v(\P)$ acts on $\varpi$ and $\b\Omega$ as it does on $\omega$ and $\Omega$ in \eqref{ActiveGT}.  The definition of local representatives via sections of $\P$, local gauge transformations and gluing properties thereof proceeds as in the standard case.

\subsection{The BRST framework}  
\label{The BRST framework}  

The infinitesimal version of \eqref{ActiveGT} can be captured by the so-called BRST differential algebra. 
 Abstractly \cite{Dubois-Violette1987} it is a   bigraded  differential algebra generated by  $\{\omega, \Omega, v, \chi\}$ where $v$ is the so-called ghost and the generators are respectively of degrees $(1, 0)$, $(2, 0)$, $(0, 1)$ and $(1, 1)$.  It is endowed with two nilpotent antiderivations $d$ and $s$, homogeneous of degrees $(1, 0)$ and $(0, 1)$ respectively, with vanishing anticommutator: $d^2=0=s^2$, $sd+ds=0$.  The algebra is equipped with a bigraded commutator $[\alpha, \beta]:=\alpha\beta -(-)^{\text{deg}[\alpha]\text{deg}[\beta]}\beta\alpha$. Notice that if the commutator vanishes identically, the BRST algebra is a  bigraded commutative differential  algebra.  
The action of $d$ is defined on the generators by: $d\omega=\Omega -\tfrac{1}{2}[\omega, \omega]$ (Cartan structure equation), $d\Omega=[\Omega, \omega]$ (Bianchi identity), $dv=\chi$ and $d\chi=0$. The action of the BRST operator on the generators gives the usual defining relations of the BRST algebra,
\begin{align}
\label{BRST}
s\omega=-dv -[\omega, v], \quad s\Omega=[\Omega, v], \quad  \text{ and } \quad sv=-\tfrac{1}{2}[v, v]. 
\end{align}
 Defining the degree $(1, 1)$ homogeneous antiderivation $\t d:=d+s$ and so-called algebraic connection $\t \omega:=\omega + v$,  \eqref{BRST} can be compactly rewritten as $\t \Omega:= \t d \t \omega + \tfrac{1}{2}[\t \omega, \t \omega]=\Omega$.  This is known as the ``russian formula''  \cite{Stora1984, Manes-Stora-Zumino1985} or ``horizontality condition'' \cite{Baulieu-TMieg1984, Baulieu-Bellon1986}. 
One is free to supplement this algebra with an element $\vphi$ of degrees $(0, 0)$  supporting a linear representation $\rho_*$ of the algebra as well as the action of the antiderivations,  so that upon defining $D:=d + \rho_*(\omega)$ one has consistently $D^2\vphi=\rho_*(\Omega)\vphi$ and 
\begin{align}
\label{BRST2}
s\vphi=-\rho_*(v)\vphi, \quad \text{ and } \quad sD\vphi=-\rho_*(v)D\vphi.
\end{align}

When the abstract BRST algebra is realized in the above differential geometric setup, the bigrading is according to the de Rham form degree and  ghost degree, $d$ is the de Rham differential on $\P$ (or $\M$) and $s$ is the de Rham operator on $\H$. The ghost is the Maurer-Cartan form on $\H$ so that $v \in \Lambda^1(\H, \text{Lie}\H)$, and given $\xi \in T\H$, $v(\xi) :\P \rarrow \LieH  \in \text{Lie}\H$ \cite{Bonora-Cotta-Ramusino}.  So in practice the ghost can be seen as a map $v:  \P\rarrow \LieH \in \text{Lie}\H$, a place holder that takes over the role of the infinitesimal gauge parameter. Thus the first two  relations of \eqref{BRST} and \eqref{BRST2} reproduce the infinitesimal gauge transformations of the gauge fields \eqref{ActiveGT}, while the third equation in \eqref{BRST} is the Maurer-Cartan structure equation for the gauge group $\H$. 

The BRST framework provides an algebraic way to characterize relevant quantities in gauge theories, such as admissible Lagrangian forms, observables and anomalies. Quantities of degree $(r, g)$ that are $s$-closed, that is $s$-cocycles $\in Z^{r,g}(s):=\ker s$, are gauge invariant. Quantities of degree $(r, g)$ that  are $s$-exact are $s$-coboundaries $\in B^{r,g}(s):=\text{Im } s$. Since $s^2=0$ obviously $B^{r,g}(s) \subset Z^{r,g}(s)$ and one defines the $s$-cohomology group $H^{r, g}(s):=Z^{r,g}(s)/B^{r,g}(s)$, elements of which  differing only by a coboundary, $c'=c+sb$, define the same cohomology class. Non-trivial Lagrangians and observables  must belong to $H^{n, *}(s)$.\footnote{If suitable boundary conditions are imposed on the fields of the theory or if the spacetime manifold is boundaryless,  the requirement of quasi-invariance of the Lagrangian, $sL=d\alpha$, is enough to ensure the invariance of the action, $\S=\int L$. So that one may consider $H^{r,g}(s|d)$, the $s$-modulo-$d$-cohomology instead of the strict $s$-cohomology.} For example, given a properly gauge invariant Yang-Mills Lagrangian $L$, $sL=0$, the prototypical Faddeev-Popov gauge-fixed Lagrangian is $L'=L+sb$, where $b$ is of degree $(n, -1)$ (since it involves an antighost, not treated here), and both belong to the same $s$-cohomology class in $H^{n, 0}(s)$. Wess-Zumino consistent gauge anomalies $\mathsf A$ - quantum gauge symmetry breaking of the quantum action $W=e^{iS}$, $sW=\mathsf A$ - belong to $H^{n,1}(s)$.

\section{Reduction of gauge symmetries: the dressing field method}   
\label{Reduction of gauge symmetries: the dressing field method}   

As insightful as the gauge principle is, gauge theories suffer from \emph{prima facie} problems such as an ill-defined quantization procedure due to the divergence of their path integral, and the masslessness of the interaction mediating fields (at odds with the phenomenology of the weak interaction). These drawbacks are rooted in the very thing that is the prime appeal of gauge theories: the gauge symmetry. Hence the necessity to come-up with strategies to reduce it. Broadly, two standard strategies to do so, addressing either problems respectively, are gauge fixings and SSB mechanisms. 
Furthermore, similarly to what happens in General Relativity (GR), it may not be straightforward to extract physical observables in gauge theories. In GR, observables must be diffeomorphism-invariant. In gauge theories, observables must be gauge-invariant, e.g: the abelian (Maxwell-Faraday) field strength or Wilson loops.

The dressing field approach is a third way, besides gauge fixing and SSB, to systematically reduce gauge symmetries. As such it may dispense to fix a gauge, can be a substitute to SSB (see \cite{Masson-Wallet,GaugeInvCompFields, Francois2014}) and provides candidate physical observables. 

\subsection{Composite fields}  
\label{Composite fields}  

Let $\P(\M, H)$ be a principal bundle equipped with a connection $\omega$ with curvature $\Omega$, and let $\vphi$ be a $\rho$-equivariant map on $\P$ to be considered as a section of the associated vector bundle $E=\P \times_H V$.  The gauge group is $\H \simeq \Aut_v(\P)$. The main content of the dressing field approach as a gauge symmetry reduction scheme is in the following

\begin{prop} 
\label{P1}
 If $K$ and $G$ are subgroups of $H$ such that $K\subseteq G \subset H$. Note $\K\subset \H$  the gauge subgroup associated with $K$. Suppose there exists a map
\begin{align} 
\label{DF}
u:\P \rarrow G \quad \text{ defined by its $K$-equivariance property }\quad  R_k^*u=k\-u,
\end{align}
  This map $u$, that we will call a \emph{dressing field}, allows to construct through $f: \P \rarrow \P$ given by $f(p)=pu(p)$, the following \emph{composite fields}
 \begin{align}
\label{CompFields}
\omega^u:&=f^*\omega=u\-\omega u+u\-du, \qquad \Omega^u:=f^*\Omega=u\-\Omega u =d\omega^u+\tfrac{1}{2}[\omega^u, \omega^u],\notag\\[1mm]
\vphi^u:&=f^*\vphi= \rho(u\-)\vphi \qquad \text{and} \qquad D^u\vphi^u:=f^*D\vphi=\rho(u\-)D\vphi=d\vphi^u+\rho_*(\omega^u)\vphi^u. 
\end{align}
 which are $\K$-invariant, $K$-horizontal and thus project on the quotient subbundle $\P/K \subset \P$.
 \end{prop}

 \begin{proof}  
 The $\K$-invariance of the composite fields  \eqref{CompFields} is most readily proven. Indeed from the definition \eqref{DF} one has $f(pk)=f(p)$ so that  $f: \P \rarrow \P/K \subset \P$ and given $\Phi(p)=p\gamma(p)$ with $\gamma \in \K \subset \H$, one has $\Phi^*f^*=(f\circ \Phi)^*=f^*$. 
 Before proving the $K$-horizontality, let us work out the expressions of the composite fields. 
 \smallskip 
 
 Let $X \in TP$ be a vector field  with flow $\phi_t$, $t\in \doubleR$ and $\phi_0=p$. The pushforward of $X_p \in T_p\P$  under $f$ is
$ f_*X_p:=\tfrac{d}{dt}f( \phi_t )|_{t=0}=\tfrac{d}{dt} \phi_t u( \phi_t )|_{t=0}=\tfrac{d}{dt} \phi_t|_{t=0} u(p) + p \tfrac{d}{dt}u(\phi_t)|_{t=0}= R_{u(p) *} X_p + p \tfrac{d}{dt}u(\phi_t)|_{t=0} $.
Now  the $\LieG=$ Lie$G$-valued $1$-form $u\-du$ is such that: $ [u\-du]_p(X_p)=u(p)\- (X\cdot u)(p) = u(p)\-\tfrac{d}{dt}  u(\phi_t)|_{t=0} \in \LieG \subset \LieH$. The associated vertical vector field at the point $f(p)$ is: 
$\big\{ [u\-du]_p(X_p)  \big\}^v_{pu(p)}:=\tfrac{d}{ds}  pu(p) e^{s[u\-du]_p(X_p)}|_{s=0}=p\tfrac{d}{dt} u(\phi_t)|_{t=0}$. Hence the final expression for the pushforward, $ f_*X_p=R_{u(p)*} X_p  +\big\{ [u\-du]_p(X_p)  \big\}^v_{pu(p)}$.
This allows to easily find the pullback of the connection and its curvature:
\begin{align*}
(f^*\omega)_p(X_p)&=\omega_{pu(p)}(f_*X_p)=\omega_{pu(p)}\big( R_{u(p)*}X_p + \big\{ [u\-du]_p(X_p)  \big\}^v_{pu(p)}\big)=R^*_{u(p)}\omega_{pu(p)}(X_p) + [u\-du]_p(X_p), \\
						&=\big(  \Ad_{u\-}\omega + u\-du  \big)_p ( X_p).\\
(f^*\Omega)_p(X_p, Y_p)&=\Omega_{pu(p)}(f_*X_p, f_*Y_p)= \Omega_{pu(p)}( R_{u(p)*}X_p,  R_{u(p)*}Y_p )=R^*_{u(p)}\Omega_p(X_p, Y_p)
						=\big(  \Ad_{u\-}\Omega \big)_p(X_p, Y_p).
\end{align*}
Clearly the Cartan structure equation holds between $f^*\Omega$ and $f^*\omega$. The pullback of $\vphi$ by $f$ is easily found to be $(f^*\vphi)(p)=\vphi(f(p))=\vphi(pu(p))=\rho(u(p)\-)\vphi(p)=(\rho(u\-)\vphi)(p)$. The result for $f^*D\vphi$ goes similarly. 
\medskip

\noindent \textbf{NB}:  The dressing field can be \emph{equally defined} by its $\K$-gauge transformation: $u^\gamma=\gamma\-u$, with $\gamma \in \K\subset \H$. Indeed, given $\Phi$ associated to  $\gamma \in \K$ and  \eqref{DF} : $(u^\gamma)(p):=\Phi^*u(p)=u(\Phi(p))=u(p\gamma(p))=\gamma(p)^{-1}u(p)=(\gamma^{-1}u)(p)$.  This together with \eqref{ActiveGT} makes easy to check algebraically that the composite fields \eqref{CompFields} are $\K$-invariant indeed,  according to $(\chi^u)^\gamma=(\chi^\gamma)^{u^\gamma}=(\chi^\gamma)^{\gamma\-u}=\chi^u$.\footnote{We use $\chi=\{\omega, \Omega, \vphi, \ldots\} $ to denote a generic variable when performing an operation that applies equally well to any specific one.}
\medskip

The $K$-horizontality of $\vphi^u$ as a $0$-form is trivial. So it is for $\Omega^u$ since $\Omega$ is tensorial. To prove it for $\omega^u$ and $D^u\vphi^u$ requires some writing. First, given $X^v \in V\P$ generated by $X \in \LieK=$ Lie$K \subset \LieH$,
\begin{align*}
\omega^u(X^v)&=u\-\omega(X^v)u+u\-du(X^v)=u\-Xu+u\-X^v\cdot u = u\-Xu+u\-\tfrac{d}{dt}u\big( p e^{tX} \big)|_{t=0},\\
				&= u\-Xu+u\-\tfrac{d}{dt}e^{-tX}u( p  )|_{t=0}= u\-Xu+u\-(-X)u=0.
\end{align*}
 Then, $D^u\vphi^u(X^v)= d\vphi^u(X^v)+ \rho_*\big( \omega^u(X^v) \big)\vphi^u=\tfrac{d}{dt} \vphi^u\left(pe^{tX}\right)|_{t=0}=\tfrac{d}{dt} \vphi^u(p)|_{t=0}=0$.
The composites fields \eqref{CompFields} are then $\K$-invariant and $K$-horizontal and project on the quotient subbundle $\P/K \subset \P$.
\end{proof}

Several comments are in order. First, in the event that $G\supset H$ then one has to assume that the $H$-bundle $\P$ is a subbundle of a $G$-bundle, and \emph{mutatis mutandis} the proposition still holds. Such a situation occurs when $\P$ is a reduction of a frame bundle  (of unspecified order) as the main object of this paper will illustrate. 

Second, if $K=H$ then the composited fields \eqref{CompFields} are $\H$-invariant, the gauge symmetry is fully reduced, and they live on $\P/H \simeq \M$. This shows that the existence of a global dressing field is a strong constraint on the topology of the bundle $\P$: a $K$-dressing field means that the bundle is trivial along the $K$-subgroup, $\P \simeq \P/K \times K$, while a $H$-dressing field means its triviality, $\P\simeq \M \times H$. 

Notice that despite the formal similarity with \eqref{ActiveGT} (or \eqref{PassiveGT}), the composite fields \eqref{CompFields} are not gauge transformed fields. Indeed the defining equivariance property \eqref{DF} of the dressing field implies $u\notin \H$, and $f\notin \Aut_v(\P)$. As a consequence, in general the composite fields do not belong to the gauge orbits of the original fields: $\chi^u \notin \mathcal{O}(\chi)$. The dressing field method then shouldn't be mistaken for a mere gauge fixing. 
\medskip

\subsection{Residual gauge symmetry}  
\label{Residual gauge symmetry}  

 Since in general $H/K$ is a coset, its action on the dressing field $u$ is left unspecified and may depend heavily on specifics of the situation at hand. Then in general nothing can be said of the transformation properties of the composite fields under $H/K$. But  interesting things happen if $K$ is a normal subgroup, $K\mathrel{\unlhd} H$, so that $H/K$ is a group that we note $J$ for convenience. The quotient bundle $\P/K$ is then a $J$-principal bundle noted $\P'=\P'(\M, J)$. 
We discuss two most important such cases in the following  subsections.
 
 \subsubsection{The composite fields as genuine gauge fields}  
 \label{The composite fields as genuine gauge fields}  

\begin{prop}  
\label{Prop2}
Let $u$ be a $K$-dressing field on $\P$. Suppose its $J$-equivariance  is given by 
\begin{align}
\label{CompCond}
R^*_ju=\Ad_{j\-}u, \qquad \text{ with } j \in J. 
\end{align}
Then the dressed connection $\omega^u$ is  a $J$-principal connection on $\P'$. That is, for $ X\in \LieJ$ and  $j\in J$, $\omega^u$ satisfies: $\omega^u(X^v)=X$ and $R^*_j\omega^u=\Ad_{j\-}\omega^u$. Its curvature is given by $\Omega^u$. 

\noindent Also, $\vphi^u$ is a $(\rho, V)$-tensorial map on $\P'$ and can be seen as a section of the associated bundle $E'=\P' \times_{J} V$. The covariant derivative on such sections is given by $D^u=d + \rho(\omega^u)$. 
\end{prop}

\begin{proof} 
Proving $\omega^u$ to be  a connection involves quite straightforward calculations. First,
\begin{align*}
\omega_p^u(X^v_p)	&=u(p)\-\omega_p(X^v_p)u(p) + u(p)\-du_p(X^v_p)= u(p)\-Xu(p) + u(p)\- (X^v\cdot u)(p)\\
					&= u(p)\-Xu(p) + u(p)\- \tfrac{d}{dt} u\left(pe^{tX}\right)|_{t=0}= u(p)\-Xu(p) + u(p)\- \tfrac{d}{dt} \left(e^{-tX}u(p)e^{tX} \right) |_{t=0}\\
					& =u(p)\-Xu(p) + u(p)\-(-X)u(p) + X=X. 
\end{align*}
 Then one finds, 
\begin{align*}
(R^*_j\omega^u)_p(X^v_p)	&=\omega^u_{pj}(R_{j*}X^v_p)=
u(pj)\- \omega^u_{pj}(R_{j*}X^v_p)\ u(pj) + u(pj)\- du_{pj}(R_{j*}X^v_p)\\
							& = u(pj)\- \Ad_{j\-}X\  u(pj) + j\-u(p)\-j \tfrac{d}{dt} u\left(p e^{tX}j\right)|_{t=0}\\
							&=j\-u(p)\- Xu(p)j + j\-u(p)\-j \tfrac{d}{dt} \left(j\-e^{-tX}u(p) e^{tX}j\right)|_{t=0}\\
							&=j\-u(p)\- Xu(p)j + j\-u(p)(-X)u(p)j + j\-u(p)\-j\ j\-u(p)Xj= j\-Xj=(\Ad_{j\-}\omega^u)_p(X^v_p).
\end{align*}
Which allows to conclude. Now since in any event the Cartan structure equation holds between $\Omega^u$ and $\omega^u$, if the latter is a genuine connection the former is its curvature.  

As for $\vphi^u$ we have 
\begin{align*}
(R^*_j\vphi^u)(p)  	&=\vphi^u(pj)=\big(\rho(u\-)\vphi\big)(pj)=\rho(u(pj)\-)\vphi(pj)=\rho(j\-u(p)\-j) \rho(j\-)\vphi(p)=\rho(j\-)\rho(u\-)\vphi(p)\\
					&=\left(\rho(j\-)\vphi^u)\right)(p).
\end{align*}
So that $\vphi^u$ is indeed a $(\rho, V)$-equivariant map on $\P'$ and a section of the associated bundle $E'=\P'\times_J V$. The covariant derivative being $D^u=d + \rho(\omega^u)$ is then standard. 
\end{proof}

From this we immediately deduce the following

\begin{cor}
\label{Cor3}
The transformation of the composite fields under the residual $\J$-gauge symmetry is found in the usual way to be
\begin{align}
\label{GTCompFields}
(\omega^u)^{\gamma'}:&={\Phi'}^*\omega^u={\gamma'}\- \omega^u \gamma' + {\gamma'}\-d\gamma',  \qquad (\Omega^u)^{\gamma'}:={\Phi'}^*\Omega^u={\gamma'}\- \Omega^u
 \gamma', \notag\\
 (\vphi^u)^{\gamma'}:&={\Phi'}^*\vphi^u=\rho({\gamma'}\- )\vphi^u, \qquad \text{ and } \qquad  (D^u\vphi^u)^{\gamma'}:={\Phi'}^*D^u\vphi^u=\rho({\gamma'}\-) D^u\vphi^u,
\end{align}
with $\Phi' \in \Aut(\P') \simeq \J \ni \gamma'$. 
\end{cor}
\begin{proof}
As said, the result follows in the usual geometric way. But there is a more algebraic derivation. For $\Phi' \in \Aut_v(\P')$ one has, using \eqref{CompCond},  $\left(u^{\gamma'}\right)(p):=({\Phi'}^*u)(p)=u(\Phi'(p))=u(p\gamma'(p))={\gamma'}(p)\-u(p)\gamma'(p)=({\gamma'}\- u \gamma')(p)$. 
 So, using again the generic variable $\chi$ one finds that $(\chi^u)^{\gamma'}=(\chi^{\gamma'})^{u^{\gamma'}}=(\chi^{\gamma'})^{{\gamma'}\-u \gamma'}=\chi^{u\gamma'}$. Which proves \eqref{GTCompFields}. 
\end{proof}

\noindent \textbf{NB}: The relation $u^{\gamma'}={\gamma'}\- u \gamma'$ can be taken as an alternative to \eqref{CompCond} as a  condition on the dressing field $u$.

\paragraph{Further dressing operations} In the case where \eqref{CompCond} holds so that the composite/dressed fields \eqref{CompFields} are $\K$-invariant but genuine $\J$-gauge fields with residual gauge transformation given by \eqref{GTCompFields}, the question stands as to the possibility to perform a further dressing operation. 

Suppose a second dressing field $u'$ for the residual symmetry is available. It would be defined by ${u'}^{\gamma'}={\gamma'}\-u'$ for $\gamma'\in \J$. But in order to not spoil the $\K$-invariance obtained from the first dressing field $u$, the second dressing field should satisfy the \emph{compatibility condition} 
\begin{align}
\label{CompCond12}
R^*_ku'=u', \quad \text{ for } k\in K. \quad \text{ Or altenatively: } \quad {u'}^\gamma=u',\quad \text{ for } \gamma \in \K.
\end{align}
In this case indeed: 
\begin{align*}
\left(\chi^{uu'}\right)^\gamma&=\left(\chi^\gamma\right)^{u^\gamma {u'}^\gamma}=\left(\chi^\gamma\right)^{\gamma\-u u'}=\chi^{uu'}, \quad \gamma\in \K.\\
\left(\chi^{uu'}\right)^{\gamma'}&=\left(\chi^{\gamma'}\right)^{u^{\gamma'} {u'}^{\gamma'}}=\left(\chi^{\gamma'}\right)^{{\gamma'}\-u\gamma' \ {\gamma'}\-u'}=\chi^{uu'}, \quad \gamma' \in \J .
\end{align*}
We see that the defining properties of the dressing fields $u$ and $u'$, together with their compatibility conditions \eqref{CompCond} and \eqref{CompCond2} implies that $uu'$ can be treated as a single dressing for $\H$: 
\begin{align*}
\left(uu'\right)^{\gamma\gamma'}=\left(\left(uu'\right)^\gamma\right)^{\gamma'}=(\gamma\- uu')^{\gamma'}=\left(\gamma^{\gamma'}\right)\-\ {\gamma'}\-u\gamma'\ {\gamma'}\-u'={\gamma'}\-\gamma\- uu'= (\gamma\gamma')\- uu'.
\end{align*}
The extension of this scheme to any number of dressing field is straightforward, the details can be found in \cite{Francois2014}.

 \subsubsection{The composite fields as a new kind of gauge fields}  
 \label{The composite fields as  a new kind of gauge fields}  
 
 Before turning to this next case we need to introduce some definitions. Let $G' \supset G$ be a Lie group for which representations $(\rho, V)$ of $G$ are also representations of $G'$.
 Let there be a $C^\infty$-map $C :P \times J \rarrow G'$, $(p, j) \mapsto C_p(j)$, satisfying
\begin{align}
\label{PropDefC}
C_p(jj')=C_p(j)C_{pj}(j').
\end{align} 
From this we have that $C_p(e)=e$, for $e$ the identity in both $J$ and $G'$, and $C_p(j)\-=C_{pj}(j\-)$. Its differential is 
\begin{align*}
dC_{(p, j)}=dC(j)_p + dC_{p| j}: T_p\P \oplus T_jJ \rarrow T_{C_p(j)}G',
\end{align*}
where $\ker dC(j)=T_jJ$ and $\ker dC_p=T_p\P$ with by definition
\begin{align*}
dC(j)_p(X_p)&=\tfrac{d}{dt}C_{\phi_t}(j) |_{t=0}, \qquad \phi_t \text{ the flow of } X\in T\P \text{ and } \phi_{t=0}=p,\\
dC_{p|j}(Y_j)&=\tfrac{d}{dt} C_p(\vphi_t)|_{t=0},\qquad \vphi_t \text{ the flow of } Y\in TJ \text{ and } \vphi_{t=0}=j.
\end{align*}
Notice that $C_p(j)\-dC_{(p,j)} : T_p\P \oplus T_jJ \rarrow  T_eG'=\LieG'$.
\medskip

We are now ready to state our next result as the following
\begin{prop}
\label{P4}
Let $u$ be a $K$-dressing field on $\P$. Suppose its $J$-equivariance is given by
\begin{align}
\label{CompCond2}
(R_j^*u)(p)=j\-u(p)C_p(j),\quad \text{with } \quad j\in J \text{ and  $C$ a map as above}. 
\end{align}
Then $\omega^u$ satisfies
\begin{enumerate}
\item $\omega^u_p(X^v_p)=c_p(X):=\tfrac{d}{dt}C_p(e^{tX})|_{t=0}$, $\quad$ for $X\in \LieJ$ and $X^v_p \in V_p\P'$.
\item $R^*_j\omega^u=C(j)\- \omega^u C(j)+ C(j)\-dC(j)$.
\end{enumerate}
So $\omega^u$ is  a kind of generalized connection $1$-form. Its  curvature $\Omega^u$ is $J$-horizontal and satisfies $R^*_j\Omega^u=C(j)\-\Omega^uC(j)$. 
Also, $\vphi^u$ is a $\rho(C)$-equivariant map, $R^*_j \vphi^u=\rho\left(C(j)\right)\- \vphi^u$. The first order differential operator $D^u:=d + \rho_*(\omega^u)$ is a natural covariant derivative on such $\vphi^u$ so that $D^u\vphi^u$ is a $(\rho(C), V)$-tensorial form: 
$R^*_jD^u\vphi^u=\rho\left(C(j)\right)\- D^u\vphi^u$ and $(D^u\vphi^u)_p(X_p^v)=0$. 
\end{prop}

\begin{proof}
Consider $X^v \in V\P'$,
\begin{align*}
\omega^u_p(X^v_p)	&=u(p)\- \omega_p(X^v_p)u(p) + u(p)\-\tfrac{d}{dt} u(pe^{tX})|_{t=0}=u(p)\-Xu(p) + u(p)\-\tfrac{d}{dt} e^{-tX}u(p)C_p(e^{tX})|_{t=0}, \\
					&= u(p)\-Xu(p)+ \left( u(p)\-(-X)e^{tX}u(p)C_p(e^{tX}) + u(p)\-e^{tX}u(p)\tfrac{d}{dt}C_p(e^{tX}) \right)\big|_{t=0},\\
					&=\tfrac{d}{dt}C_p(e^{tX})|_{t=0}=:c_p(X).
\end{align*}
Take now $X_p \in T_p\P'$ with flow $\phi_t$, 
\begin{align*}
(R^*_j\omega^u)_p(X_p)&=\omega^u_{pj}(R_{j*}X_p)= u(pj)\- \omega_{pj}(R_{j*}X_p)) u(pj) + u(pj)\-\tfrac{d}{dt} u(\phi_t j) |_{t=0}, \\
						&=C_p(j)\-u(p)\-j \Ad_{j\-}\omega_p(X_p) j\-u(p)C_p(j) +u(pj)\-\tfrac{d}{dt} j\- u(\phi_t)C_{\phi_t}(j) |_{t=0},\\
						&=C_p(j)\-u(p)\omega_p(X_p) u(p)C_p(j) +C_p(j)\-u(p)\-\tfrac{d}{dt} u(\phi_t)|_{t=0}C_{p}(j)  + C_p(j)\-\tfrac{d}{dt} C_{\phi_t}(j) |_{t=0},\\
						&=C_p(j)\-u(p)\omega_p(X_p) u(p)C_p(j) +C_p(j)\-u(p)\-du_p(X_p) C_p(j) + C_p(j)\- dC(j)_p(X_p),\\
						&=\left( C(j)\- \omega^u C(j) +C(j)\-dC(j) \right)_p(X_p).
\end{align*}
The $J$-horizontality of $\Omega^u$ follows from that of $\Omega$. Its equivariance property is proved either in a way similar as above, or from the Cartan structure equation and using the just proved equivariance of $\omega^u$. The equivariance of $\vphi^u$ is found in a one line calculation:
$(R^*_j\vphi^u)(p)=\vphi^u(pj)=\rho(u(pj))\-\vphi(pj)=\rho\left( C_p(j)\- u(p)\-j \right)\rho(j\-)\vphi(p)=\rho( C_p(j)\-)\vphi^u(p)$. From this and the two properties of $\omega^u$  we easily prove that $D^u$ is indeed a covariant derivative for $\vphi^u$. On the one hand the equivariance:
\begin{align*}
(R^*_jD^u\vphi^u)_p(X_p) &=d\vphi_{pj}(R_{j*}X_p) + \rho_*( \omega^u_{pj})(R_{j*}X_p)\vphi^u(pj),\\
						&=\tfrac{d}{dt} \vphi^u(\phi_tj) |_{t=0}+ \rho_*(R^*_j\omega^u_p)(X_p) \rho\left(C_p(j)\-\right) \vphi^u(p),\\
						&=\tfrac{d}{dt} \rho\left(C_{\phi_t}(j) \- \right) \vphi^u(\phi_t) |_{t=0}+ \rho_*\left( C_p(j)\-\omega^u_pC_p(j) + C_p(j)\- dC(j)_p \right)(X_p)\rho\left(C_p(j)\-\right) \vphi^u(p),\\
						=\rho_*\left(d{C(j)\-}_p\right) &(X_p) \vphi^u(p) + \rho\left(C_p(j)\-\right)d\vphi^u_p(X_p) + \rho\left( C_p(j)\-\right)\rho_*( \omega^u_p) (X_p)\vphi^u(p) - \rho_*\left( d{C(j)\-}_p\right)(X_p)\vphi^u(p) ,\\
						&=\left(\rho\left(C(j)\-\right) D^u\vphi^u \right)_p(X_p).
\end{align*}	
On the other hand, the horizontality:
\begin{align*}
(D^u\vphi^u)_p(X_p^v)&=d\vphi^u_p(X^v_p) + \rho_*(\omega^u_p)(X^v_p)\vphi^u(p)=\tfrac{d}{dt} \vphi^u\left(pe^{tX}\right)|_{t=0} +\rho_*\left(c_p(X)\right)\vphi^u(p),\\
					=\tfrac{d}{dt}& \rho\left(C_p\left(e^{tX}\right)\-\right)\big|_{t=0} \vphi^u(p) + \tfrac{d}{dt} \rho\left(C_p\left(e^{tX}\right)\right)\big|_{t=0} \vphi^u(p)=\tfrac{d}{dt} \rho\left(C_p\left(e^{tX}\right)\-C_p\left(e^{tX} \right) \right)\big|_{t=0} \vphi^u(p)=0.
\end{align*}
\end{proof}

From this we can find the transformations of the composite fields under the residual gauge group $\J \simeq \Aut_v(\P')$. But first, we again need some preliminary results. Consider $\Phi' \in \Aut_v(\P')\simeq  \gamma' \in \J$, the residual gauge transformation of the dressing field is
\begin{align}
\label{GT-CDressField}
\left(u^{\gamma'} \right)(p):=(\Phi'^*u)(p)&=u(p\gamma'(p))=\gamma'(p)\- u(p) C_p\left(\gamma'(p)\right)=\left( {\gamma'}\- u C(\gamma') \right) (p).
\end{align}
\textbf{NB}: This relation can be taken as an alternative to \eqref{CompCond2} as a  condition on the dressing field $u$.\\
We witness the introduction of the map $C(\gamma'):\P' \rarrow G'$, $p \mapsto C_p\left(\gamma'(p)\right)$. It is given by the composition serie
\begin{align*}
&\P' \xrightarrow{\Delta} \P' \times \P' \xrightarrow{\id \times \gamma'} \P' \times J \xrightarrow{\phantom{b}C\phantom{b}} G', \\
& p\xmapsto{\phantom {bbbb}} (p, p) \xmapsto{\phantom{bbbb}}  \left(p, \gamma'(p)\right) \xmapsto{\phantom{bb}} C_p\left(\gamma'(p)\right).
\end{align*}
Its differential  $dC(\gamma') : T_p\P' \rarrow T_{C_p\left(\gamma'(p)\right)}G'$ is given by $dC(\gamma')=dC \circ \left( \id \oplus d\gamma' \right) \circ d\Delta$. So given  $X_p \in T_p\P'$ with flow $\phi_t$ and $d\gamma'_p(X_p)\in T_{\gamma(p)}J$, we have explicitly
\begin{align}
\label{DiffC(gamma)}
dC(\gamma')_p(X_p)&=dC\left(\gamma'(p)\right)_{|p} \oplus dC_{p | \gamma'(p)} \left(X_p + d\gamma'_p(X_p) \right)=dC\left(\gamma'(p)\right)_{|p}(X_p) + dC_{p | \gamma'(p)} \left(d\gamma'_p(X_p) \right),\notag\\
				&= dC\left(\gamma'(p)\right)_{|p}(X_p) + dC_p(\gamma')_{|p}(X_p) = \tfrac{d}{dt} C_{\phi_t}\left( \gamma'(p) \right) |_{t=0} + \tfrac{d}{dt} C_p\left( \gamma'(\phi_t) \right) |_{t=0} .
\end{align}
Notice that $C_p\left(\gamma'(p) \right)\-dC(\gamma')_p : T_p\P' \rarrow T_eG'=\LieG'$. 

\medskip We are now ready to give the transformations of the composite fields under the residual gauge group. 
\begin{prop} 
\label{P5}
Given  $\Phi' \in \Aut_v(\P')\simeq  \gamma' \in \J$, the residual gauge transformations of the composite fields are
\begin{align}
\label{GTCompFields2}
(\omega^u)^{\gamma'}:&= \Phi'^*\omega^u = C(\gamma')\- \omega^u C(\gamma') + C(\gamma')\-dC(\gamma'), \qquad (\Omega^u)^{\gamma'}:=\Phi'^*\Omega^u= C(\gamma') \- \Omega^u C(\gamma'), \notag \\
(\vphi^u)^{\gamma'}:&= \Phi'^*\vphi^u = \rho\left( C(\gamma') \- \right) \vphi^u \qquad \text{ and } \qquad (D^u\vphi^u)^{\gamma'}= \Phi'^* D^u\vphi^u = \rho\left( C(\gamma') \- \right) D^u\vphi^u.
\end{align}
So, the composite fields \eqref{CompFields} behave as gauge fields of a new kind, and implement the \emph{gauge principle} - or principle of \emph{local symmetry} - of field theory in Physics. 
\end{prop}

\begin{proof} 
The pushforward of $X_p \in T_p\P'$ under $\Phi' \in \Aut_v(\P')\simeq  \gamma' \in \J$ is  $\Phi'_*X_p=R_{\gamma'(p)*}X_p + \left\{ [{\gamma'}\-d\gamma']_p(X_p) \right\}\big|^v_{\Phi'(p)}$. So the pullback of $\omega^u$ is
\begin{align*}
(\Phi'^*\omega^u)_p(X_p)=\omega^u_{\Phi'(p)}(\Phi'_*X_p)&=\omega^u_{\Phi'(p)}\left(R_{\gamma'(p)*}X_p + \left\{ [{\gamma'}\-d\gamma']_p(X_p) \right\}\big|^v_{\Phi'(p)}\right),\\
													&=\left(R^*_{\gamma'(p)}\omega^u\right)_p(X_p) + c_{p\gamma'(p)}\left(  [{\gamma'}\-d\gamma']_p(X_p)  \right).
\end{align*}
Now by definition of $c_p$ and using \eqref{PropDefC}
\begin{align*}
c_{p\gamma'(p)}\left(  [{\gamma'}\-d\gamma']_p(X_p)  \right):&=\tfrac{d}{dt} C_{p\gamma'(p)}\left( \gamma'(p)\- \gamma'(\phi_t) \right) \big|_{t=0}=\tfrac{d}{dt} C_p(\gamma'(p))\- C_p\left(  \gamma'(\phi_t) \right) \big|_{t=0},\\
														&= C_p(\gamma'(p))\- \tfrac{d}{dt} C_p\left(  \gamma'(\phi_t) \right) \big|_{t=0}= C_p(\gamma'(p))\-dC_p(\gamma')_{|p}(X_p).
\end{align*}
Hence, using \eqref{DiffC(gamma)} to conclude, we get
\begin{align*}
(\Phi'^*\omega^u)_p(X_p)&=\left( C\left(\gamma'(p)\right)\- \omega^u C\left(\gamma'(p)\right) +C\left(\gamma'(p)\right)\-dC(\gamma'(p)) \right)_p(X_p) + C_p(\gamma'(p))\-dC_p(\gamma')_{|p}(X_p),\\
						&=\left( C\left(\gamma'\right)\- \omega^u C\left(\gamma'\right) +C\left(\gamma'\right)\-dC(\gamma') \right)_p(X_p).
\end{align*}
The pullback of $\Omega^u$ is
\begin{align*}
(\Phi'^*\Omega^u)_p(X_p, Y_p)&=\Omega^u_{\Phi'(p)}(\phi_*X_p, \phi_*Y_p)=\Omega^u_{\Phi'(p)}(R_{\gamma'(p)*}X_p, R_{\gamma'(p)*}Y_p)=\left(R^*_{\gamma'(p)}\Omega^u\right)_p(X_p, Y_p),\\
&= \left( C(\gamma'(p))\- \Omega^u C(\gamma'(p))  \right)_p(X_p, Y_p)= \left( C(\gamma')\- \Omega^u C(\gamma')  \right)_p(X_p, Y_p).
\end{align*}
The tensoriality of $\Omega^u$ has been used. The pullback of $\vphi^u$ is easily found to be
\begin{align*}
\left(\Phi'^*\vphi^u \right)(p)=\vphi^u(p\gamma'(p))=\rho\left( C_p(\gamma'(p))\- \right)\vphi^u(p)=\left( \rho\left( C(\gamma')\- \right)\vphi^u\right)(p).
\end{align*}
The proof for $D^u\vphi^u$ goes similarly. 

\noindent \textbf{Remark}: Given the residual gauge transformation of the dressing field \eqref{GT-CDressField} and the usual $\J$-gauge transformations for the standard gauge fields $\chi$, the above geometrical proof secures the more direct algebraic calculation: $(\chi^u)^{\gamma'}=(\chi^{\gamma'})^{u^{\gamma'}}=(\chi^{\gamma'})^{{{\gamma'}\-}uC(\gamma')}=\chi^{uC(\gamma')}$.
\end{proof}

\noindent \textbf{NB}: Under a further gauge transformation $\Psi \in \Aut_v(\P') \simeq \eta\in \J$, the dressing field behaves as
\begin{align*}
\left(\Psi^*(\Phi^*u)\right)(p)&= \left( (\Phi \circ \Psi)^*u \right)(p)=u\left(\Phi(p\eta(p)\right)=u\left(\Phi(p)\eta(p)\right)=u\left(p\gamma(p)\eta(p)\right)= \eta(p)\-\gamma(p)\- u(p)C_p\left(\gamma(p)\eta(p)\right),\\
							&=\left(\eta\-\gamma\-\ u\ C\left(\gamma\eta \right)\right)(p).\\
\text{or }  \left(\Psi^*(\Phi^*u)\right)(p)&= \left( \gamma\- u C(\gamma) \right) (\Psi(p))=\gamma\left(p\eta(p) \right)\- u\left(p\eta(p)\right) C_{p\eta(p)}\left(\gamma(p\eta(p) \right),\\
							&= \eta(p)\-\gamma(p)\-\eta(p) \cdot \eta(p)\-u(p) C_p\left( \eta(p)\right) \cdot C_{p\eta(p)}\left(\eta(p)\- \gamma(p) \eta(p) \right),\\
							&=\eta(p)\-\gamma(p)\-\ u(p) C_p\left( \eta(p)\right) \cdot C_{p\eta(p)}\left(\eta(p)\- \right) C_p\left(\gamma(p) \eta(p) \right) = \eta(p)\-\gamma(p)\- u(p)C_p\left(\gamma(p)\eta(p)\right).
\end{align*}
This secures the fact that the action  \eqref{GTCompFields2} of the residual gauge symmetry on the composites fields is well-behaved as a representation.

\paragraph{Further dressing operations} In the case where \eqref{CompCond2} holds, the composite/dressed fields \eqref{CompFields} are $\K$-invariant but  $\J$-gauge fields of a new kind with  gauge transformations given by \eqref{GTCompFields2}. As such they implement the gauge principle of field theory in Physics, so the dressing field philosophy applies, as shown in the following

\begin{prop}
Suppose a $J$-dressing field $u':\P' \rarrow J$ is available. Then the map $C(u') :\P' \rarrow G'$, $p\mapsto C_p(u'(p))$ is a $C(J)$-dressing field
and its $\J$-gauge transformation is  $C(u')^{\gamma'}=C(\gamma')\- C(u')$
\end{prop}
\begin{proof} 
The $J$-dressing field $u'$ is defined by $R^*_ju'=j\- u'$ for $j\in J$. So the equivariance of $C(u')$ is
\begin{align*}
\left(R^*_jC(u')\right)(p)=C_{pj} \left( u'(pj)\right)=C_p(j)\- C_p(ju'(pj))=C_p(j)\-C_p(u'(p)) = \left(C(j)\-C(u') \right)(p).
\end{align*}
This is indeed the defining property of a $C(J)$-dressing field. Given $\Phi'\in \Aut_v(P') \simeq \gamma' \in \J$, its gauge transformation is
\begin{align*}
\left(C(u')^{\gamma'}\right)(p):=\left({\Phi'}^*C(u')\right)(p)=\left( R^*_{\gamma'(p)} C(u')\right)(p)= C_p(\gamma'(p))\- C_p(u'(p))=\left(C(\gamma')\-C(u')\right)(p).
\end{align*}
\end{proof}
 In order to not spoil the $\K$-invariance obtained from the first dressing field $u$, in addition for $u'$ to satisfy \eqref{CompCond12}, the $K$-equivariance of the map $C_p$ should be trivial: $R^*_kC_p=C_p$. 
In this case indeed the $C(J)$-dressing is $K$-invariant: $\left(R^*_kC(u')\right)(p) =C_{pk}(u'(pk))=C_p(u'(p))=\left( C(u')\right)(p)$. So one has
\begin{align*}
\left(\chi^{uC(u')}\right)^\gamma&=\left(\chi^\gamma\right)^{u^\gamma C(u')^\gamma}=\left(\chi^\gamma\right)^{\gamma\-u C(u')}=\chi^{uC(u')}, \quad \gamma\in \K.\\
\left(\chi^{uC(u')}\right)^{\gamma'}&=\left(\chi^{\gamma'}\right)^{u^{\gamma'} C(u')^{\gamma'}}=\left(\chi^{\gamma'}\right)^{{\gamma'}\-uC(\gamma') \ C(\gamma')\-C(u')}=\chi^{uC(u')}, \quad \gamma' \in \J .
\end{align*}
The properties  of the dressing fields $u$ and $C(u')$ implies that $uC(u')$ can be treated as a single dressing for $\H$: 
\begin{align*}
\left(uC(u')\right)^{\gamma\gamma'}=\left(\left(uC(u')\right)^\gamma\right)^{\gamma'}=(\gamma\- uC(u'))^{\gamma'}=\left(\gamma^{\gamma'}\right)\-\ {\gamma'}\-uC(\gamma')\ C(\gamma')\-C(u')={\gamma'}\-\gamma\- uC(u').
\end{align*}

\paragraph{The case of  $1$-$\alpha$-cocycles} Suppose $C_p : J \rarrow G'$ is defined by $C_p(jj')=C_p(j)\ \alpha_j[C_p(j')]$, for $\alpha: J \rarrow \Aut(G')$ a continuous group morphism. Such objects appear in the representation theory of crossed products of $C^*$-algebras and is known as a \emph{$1$-$\alpha$-cocycle} (see \cite{Pedersen79,Williams07}).\footnote{In the general theory the group $G'$ is replaced by a $C^*$-algebra $A$.} Its defining property is an example of \eqref{PropDefC}, and everything that has been said in this section -and will be said in the following - applies when $C_p$ is a $1$-$\alpha$-cocycle. 
\medskip

As a particular case, consider the following
\begin{prop}
\label{alpha_cocycle}
 Suppose $J$ is abelian and let $A_p, B : J\rarrow GL_n$  be group morphisms where $R^*_jA_p(j')=B(j)\- A_p(j') B(j)$. Then $C_p:=A_pB : J\rarrow GL_n$ is a $1$-$\alpha$-cocyle where the morphism $\alpha: J \rarrow \Aut(GL_n)$ is the conjugate action through the morphism $B$: $\alpha_j [g] =B(j)\-[g] B(j)$, with $g \in GL_n$.
 \end{prop}
\begin{proof}
Using the commutativity of $J$ the proposition is proven in a one line calculation:
\begin{align*}
C_p(jj')=A_p(jj')B(jj')=A_p(j)A_p(j')B(j)B(j')=A_p(j)B(j)\ B(j)\-[ A_p(j')B(j')] B(j)=C_p(j)\ B(j)\-[C_p(j')]B(j). 
\end{align*}
Notice by the way that we have $C_p(jj')=C_p(j'j)=C_p(j')\ B(j')\-[C_p(j)]B(j')$, as is easily seen. 
\end{proof}
As a matter of fact  in the case soon to be discussed of the conformal Cartan geometry and the associated Tractors and Twistors, $1$-$\alpha$-cocycles of this type - where $J$ is the Weyl group of rescalings - are involved.

\subsection{Application to the BRST framework} 
\label{Application to the BRST framework} 

The BRST algebra encodes the infinitesimal gauge symmetry. It is to be expected that the dressing field method modifies it. To see how, let us first consider the following 
\begin{prop}    
Given the BRST algebra \eqref{BRST}-\eqref{BRST2} on the initial gauge variables and the ghost $v\in Lie\H$. The composite fields \eqref{CompFields} satisfy the  \emph{modified BRST algebra}:
\begin{align}
\label{NewBRST}
&s\omega^u=-D^uv^u=-dv^u-[\omega^u, v^u], \quad s\Omega^u=[\Omega^u, v^u], \quad s\vphi^u=-\rho_*(v^u)\vphi^u,\quad \text{and } \quad sv^u=-\tfrac{1}{2}[v^u, v^u]\\[2mm]
 &\text{with the \emph{dressed ghost} } \quad v^u=u\-vu + u\-su.  \notag
\end{align}
This result does not rest on the assumption that $u$ is a dressing field. Furthermore one defines the \emph{dressed algebraic connection} as
\begin{align*}
\t{\omega^u}=\omega^u + v^u=u\- \t\omega u + u\-\t du.
\end{align*}
\end{prop}

\begin{proof} 
The result is easily found by expressing the initial gauge variable $\chi=\{ \omega, \Omega, \vphi\}$ in terms of the dressed fields $\chi^u$ and the dressing field $u$, and re-injecting in the initial BRST algebra  \eqref{BRST}-\eqref{BRST2}. At no point of the derivation does $su$ need to be explicitly known. It then holds regardless if $u$ is a dressing field or not. 
\end{proof}

If the ghost $v$ encodes the infinitesimal initial $\H$-gauge symmetry, the dressed ghost $v^u$ encodes the infinitesimal residual gauge symmetry. Its concrete expression depends on the BRST transformation of $u$. 

Under the hypothesis $K\subset H$, the ghost decomposes as $v=v_\LieK + v_{\LieH/\LieK}$, and the BRST operator splits accordingly: $s=s_\LieK + s_{\LieH/\LieK}$. If $u$ is a dressing field its BRST transformation is the infinitesimal version of its defining transformation property: $s_\LieK u=-v_\LieK u$. So the dressed ghost is
\begin{align*}
v^u=u\-vu+u\-su=u\-(v_\LieK + v_{\LieH/\LieK})u+u\-(-v_\LieK u+ s_{\LieH/\LieK}u)=u\- v_{\LieH/\LieK} u+u\- s_{\LieH/\LieK}u. 
\end{align*}
We see that the Lie$\K$ part of the ghost, $v_\LieK$, has disappeared. This means that $s_\LieK \chi^u=0$, which expresses the $\K$-invariance of the composite fields \eqref{CompFields}. 

\paragraph{Residual BRST symmetry}  

In general $\LieH/\LieK$ is simply a vector space, so $s_{\LieH/\LieK} u$ is left unspecified and nothing can be said in general of $v^u$ and of the form of the modified BRST algebra \eqref{NewBRST}. But following  section \ref{Residual gauge symmetry}, if $\K\mathrel{\unlhd}H$ then $H/K=J$ is a group with lie algebra $\LieH/\LieK=\LieJ$. We here provide the BRST treatment of the two cases detailed in this section. 
\medskip

Suppose the dressing field satisfies the condition \eqref{CompCond}, whose BRST version is: $s_{\LieJ} u=[u, v_{\LieJ}]$. The dressed ghost is then
\begin{align}
\label{NewGhost2}
v^u=u\- v_{\LieJ} u+u\- s_{\LieJ}u=u\- v_{\LieJ} u+u\- (u v_{\LieJ} -  v_{\LieJ}u)=v_{\LieJ}. 
\end{align}
This in turn implies that the new BRST algebra is
\begin{align}
\label{NewBRST2}
s\omega^u=-D^uv_{\LieJ}=-dv_{\LieJ}-[\omega^u, v_{\LieJ}], \quad s\Omega^u=[\Omega^u, v_{\LieJ}], \quad s\vphi^u=-\rho_*(v_{\LieJ})\vphi^u,\quad \text{and } \quad sv_{\LieJ}=-\tfrac{1}{2}[v_{\LieJ}, v_{\LieJ}].
\end{align}
This is the BRST version of \eqref{GTCompFields}, and reflects the fact that the composites fields \eqref{CompFields} are genuine $\J$-gauge fields, in particular that $\omega^u$ is a $J$-connection. 

A further dressing field $u'$ would be defined by $s_{\LieJ} u'=-v_{\LieJ} u'$, and the necessary compatibility condition it needs to satisfy is $s_\LieK u'=0$. The combined dressing $uu'$ is such that $suu'=-vuu'$, so that $v^u=0$ and $s\chi^{uu'}=0$. Again the straightforward extension of the scheme to any number of dressing fields can be found in \cite{Francois2014}. 
\medskip

Suppose now that the dressing field satisfies the  condition \eqref{CompCond2}, whose BRST version is: $s_{\LieJ} u=-v_{\LieJ}u + uc_p(v_{\LieJ})$. The dressed ghost is then
\begin{align}
\label{NewGhost3}
v^u=u\- v_{\LieJ} u+u\- s_{\LieJ}u=u\- v_{\LieJ} u+u\- \left(-v_{\LieJ}u + uc_p(v_{\LieJ})\right)=c_p(v_{\LieJ}). 
\end{align}
This in turn implies that the new BRST algebra is
\begin{align}
\label{NewBRST3}
&s\omega^u=-dc_p(v_{\LieJ})-[\omega^u, c_p(v_{\LieJ})], \quad s\Omega^u=[\Omega^u, c_p(v_{\LieJ})], \quad s\vphi^u=-\rho_*(c_p(v_{\LieJ}))\vphi^u,\\[1mm]
\quad \text{and } \quad &sc_p(v_{\LieJ})=-\tfrac{1}{2}[c_p(v_{\LieJ}), c_p(v_{\LieJ})]. \notag
\end{align}
This is the BRST version of \eqref{GTCompFields2}, and reflects the fact that the composites fields \eqref{CompFields} instantiate the gauge principle in a satisfactory way.

A further dressing field $C(u')$ would be defined by $s_{\LieJ} C(u')=-c(v_{\LieJ}) C(u')$, and the  compatibility condition it needs to satisfy is $s_\LieK C(u')=0$. The combined dressing $uC(u')$ is such that $s\left( uC(u')\right)=-v\left( uC(u')\right)$, so that $v^u=0$ and $s\chi^{uC(u')}=0$.
\bigskip

\subsection{Local aspects and Physics} 
\label{Local aspects and Physics} 

 Up until now we have exposed in great details the global aspects of the dressing approach on the bundle $\P$  to emphasize the geometric nature of the composites fields obtained, according to the given equivariance properties displayed by the dressing field. Most notably we showed that  the composite field can behave as a new kind of gauge fields. 
 
But to do Physics we need the local representatives on an open subset $\U\subset \M$ of global dressing and composite fields. These are obtained in the usual way from a local section $\sigma:\U \rarrow \P$ of the bundle. The  important properties they thus retain is their gauge invariance and residual gauge transformations. 

If it happens that a dressing field is defined locally on $\U$ first, and not directly on $\P$, then the local composite fields $\chi^u$ are defined in terms of the local dressing field $u$ and local gauge fields $\chi$ by \eqref{CompFields}. The gauge invariance and residual gauge transformations of these local composite fields are derived from the gauge transformations of the local dressing field under the various subgroups of the local gauge group $\H_\text{\tiny{loc}}$ according to $(\chi^u)^\gamma=(\chi^\gamma)^{u^\gamma}$. The BRST treatment for the local objects mirrors exactly the one given for the global objects.  

This being said, note $A=\sigma^*\omega$, $F=\sigma^*\Omega$ for definiteness but keep $u$ and $\vphi$ to denote the local dressing field and section. We state the final proposition of this section, dealing with gauge theory.
\begin{prop} 
\label{Prop-Lagrangian}
Given the geometry defined by a bundle $\P(\M, H)$ endowed with $\omega$ and the associated bundle $E$, suppose we have a gauge theory given by the prototypical $\H_\text{\tiny{loc}}$-invariant Yang-Mills Lagrangian
\begin{align*}
L(A, \vphi)=\tfrac{1}{2}\Tr(F \w * F) + \langle D\vphi, \ *D\vphi\rangle - U(||\vphi||), 
\end{align*}
where $ ||\vphi||:=|\langle \vphi \rangle |^{\sfrac{1}{2}}$. If there is a local dressing field $u: \U \rarrow G  \subset H$ with $\K_\text{\tiny{loc}}$-gauge transformation $u^\gamma=\gamma\-u$, then the above Lagrangian is actually  a $\H_\text{\tiny{loc}}/\K_\text{\tiny{loc}}$-gauge theory defined in terms of  $\K_\text{\tiny{loc}}$-invariant variables since we have
\begin{align*}
L(A, \vphi)=L(A^u, \vphi^u)=\tfrac{1}{2}\Tr(F^u \w * F^u) + \langle D^u\vphi^u, \ *D^u\vphi^u\rangle - U(||\vphi^u||)
\end{align*}
by a mere change of variables. 
\end{prop}
\begin{proof}
The result follows straightforwardly from the $\H_\text{\tiny{loc}}$-invariance of the initial Lagrangian. Since $L(A^\gamma, \vphi^\gamma)=L(A, \vphi)$ for $\gamma:\U \rarrow H$,
 holds as a formal property of $L$, it follows that $L(A^u, \vphi^u)=L(A, \vphi)$ for $u:\U \rarrow G \subset H$.
\end{proof}

Notice  that since $u$ is a dressing field, $u \notin \H_\text{\tiny{loc}}$ so the dressed Lagrangian $L(A^u, \vphi^u)$ ought not to be confused with a gauge-fixed Lagrangian $L(A^\gamma, \vphi^\gamma)$ for some chosen $\gamma \in \H_\text{\tiny{loc}}$, even if it may happen that $\gamma=u$.\footnote{Remember indeed the comments at the end of section \ref{Composite fields}.}
 A fact that might go unnoticed. As we've stressed in the opening of section \ref{Reduction of gauge symmetries: the dressing field method}, the dressing field approach is distinct from both gauge-fixing and spontaneous symmetry breaking as a means to reduce gauge symmetries. 

 Let us highlight the fact that a dressing field can often be constructed  by requiring the gauge invariance of a prescribed ``gauge-like condition''. 
  Such a condition is given when a local gauge field $\chi$ (often the gauge potential)  transformed by a field $u$ with value in the symmetry group $H$, or one of its subgroups, is required to satisfy a functional constraint: $\Sigma(\chi^u)=0$.  Explicitly solved, this makes $u$ a function of $\chi$, $u(\chi)$, thus sometimes called \emph{field dependent gauge transformation}. However this terminology is valid if and only if  $u(\chi)$ transforms under the action of $\gamma\in \H_\text{\tiny{loc}}$ as $u(\chi)^\gamma:=u(\chi^\gamma)=\gamma\- u(\chi) \gamma$, in which case $u(\chi) \in \H_\text{\tiny{loc}}$.  But if the functional constraint  still holds under the action of $\H_\text{\tiny{loc}}$, or of a subgoup thereof, it follows  that $(\chi^\gamma)^{u^\gamma}=\chi^u$ (or equivalently that $s\chi^u=0$). This in turn imposes that $u^\gamma=\gamma\- u$ (or $su=-vu$) so that $u \notin \H_\text{\tiny{loc}}$ but is indeed a dressing field.

 This and the above proposition generalizes the pioneering idea of Dirac \cite{Dirac55, Dirac58}  aiming at quantizing QED by rewriting the classical theory in terms of gauge-invariant variables. The idea was rediscovered several times, early by Higgs himself \cite{Higgs66} and Kibble \cite{Kibble67}. The invariant variables were sometimes termed \emph{Dirac variables} \cite{Pervushin, Lantsman} and reappeared in various contexts in gauge theory, such as  QED \cite{Lavelle-McMullan93},  quarks theory in QCD \cite{McMullan-Lavelle97}, the proton spin decomposition controversy  \cite{LorceGeomApproach, Leader-Lorce, FLM2015_I} and most notably  in electroweak theory and Higgs mechanism \cite{Frohlich-Morchio-Strocchi81, McMullan-Lavelle95, Chernodub2008, Faddeev2009, Masson-Wallet, Ilderton-Lavelle-McMullan2010, Struyve2011, vanDam2011}. Indeed,  proposition \ref{Prop-Lagrangian} applies to the electroweak sector of the Standard Model and thus provides an alternative  to the usual textbook interpretation of the Higgs mechanism in terms of spontaneous symmetry breaking, see \cite{GaugeInvCompFields, Francois2014} for the explicit dressing field treatment. 
  
The dressing field approach thus gives a unifying and clarifying framework for these works, and others concerning the BRST treatment of anomalies in QFT \cite{Manes-Stora-Zumino1985, Garajeu-Grimm-Lazzarini}, Polyakov's ``partial gauge fixing'' for $2D$-quantum gravity \cite{Polyakov1989, Lazzarini2008} or the construction of the Wezz-Zumino functionnal \cite{Attard-Lazz2016}.
 It is the aim of this paper and  its companion to show that both tractors and twistors can also be encompassed by this approach, which furthermore highlights their nature as gauge fields of a non-standard kind. The case of tractors is dealt with in the next section. Twistors will be treated in the companion paper.

\section{Tractors from conformal Cartan geometry via dressing} 
\label{Tractors from conformal Cartan geometry via dressing} 

 Due to the tremendous progress of the last twenty years in the mathematics of parabolic geometries, the term \emph{tractor} is now more general than it used to. Given a Cartan geometry $(\P, \varpi)$ of type $(G, H)$, the reference text \cite{Cap-Slovak09}, section $1. 5.7$, defines  tractor bundles as the class of natural\footnote{\emph{Natural} is taken in the precise technical sense of \cite{Cap-Slovak09} section $1.5.5$, essentially as being associated to of higher-order frame bundles.}
  vector bundles associated to $\P$  where the action of $H$ is a restriction of an action by $G$. Connections naturally induced by the Cartan connection $\varpi$ on tractor bundles are called tractor connections. An example  is the \emph{adjoint tractor bundle} $\P\times_H \LieG$, where $H$ acts on $\LieG$ by the restriction of the adjoint action of $G$. The curvature $\b\Omega$ of $\varpi$ takes values in the sections of the adjoint tractor bundle. If $\doubleR^n$ is the defining representation of $G$, then the bundle $\P\times_H \doubleR^n$ is the \emph{standard tractor bundle}.  
  
  However, as mentioned in our introduction, initially the standard tractor bundle was devised for conformal (and projective) manifolds and constructed via prolongation of a defining differential equation. A procedure deemed at the time more  explicit than the associated bundle construction, facilitating calculations \cite{Curry-Gover2015} and  easier as a direct definition \cite{Bailey-et-al94}. This procedure we review briefly in the following section, so that the reader can compare with the derivation via the dressing field method in the next.

\subsection{Bottom-up construction via prolongation of the Almost Einstein equation} 
\label{Bottom-up construction via prolongation of the Almost Einstein equation} 

One starts with a $n$-dimensional conformal manifold $(\M, c)$ with $c$ the conformal class of the Levi-Civita connection.
Define the operator $A_{\mu\nu}:=\text{TF} \left( \nabla_\mu \nabla_\nu - P_{\mu\nu}\right)$, where TF means ``trace-free'' in the metric sense, $\nabla$ is the covariant derivative associated to a choice of metric $g \in c$ and $P_{\mu\nu}= -\tfrac{1}{n-2}\left( R_{\mu\nu} - \tfrac{R}{2(n-1)}g_{\mu\nu}\right)$ is the Schouten tensor. It is a computational exercise to show that $A_{\mu\nu}$ is a covariant operator on conformal $1$-densities $\s \in \E[1]$: $\h A_{\mu\nu} \circ z =z \circ A_{\mu\nu}$. That is, under the Weyl conformal rescaling of the metric $\h g=z^2g$, one has $\h \s=z\s$ and $\h A_{\mu\nu} \h \s=zA_{\mu\nu}\s$. Such a $1$-density $\s$ is often called a \emph{scale}, since it can be used to define a so-called conformal metric $\s^{-2}g$ representative of the conformal class $c$.\footnote{In a forthcoming note we will show how this move can be understood in  the light of the dressing field method.} The operator $A_{\mu\nu}$ is thus well-defined on $(\M, c)$. 

One then defines the so-called Almost Einstein (AE) equation on $(\M, c)$ as $A_{\mu\nu}\s=0$, explicitly
\begin{align}
\label{AE1}
\text{TF} \left( \nabla_\mu \nabla_\nu \s - P_{\mu\nu}\s \right)= \nabla_\mu \nabla_\nu \s - P_{\mu\nu} \s  - \tfrac{g_{\mu\nu}}{n} \left( \Delta\s - P \s \right)=0,
\end{align}
 with $\Delta:=g^{\mu\nu}\nabla_\mu\nabla_\nu$ and $P:=g^{\mu\nu}P_{\mu\nu}$. It is thus named because if $\s$ is a solution then the metric it determines is Einstein \cite{Bailey-et-al94, Curry-Gover2015}. 
This is the differential equation to be prolonged and recast as a system of first-order differential equations.
To do so one defines the intermediary variables $\ell_\nu=\nabla_\nu \s$ and $\rho=- \tfrac{1}{n} \left( \Delta\s - P \s \right)$, so that \eqref{AE1} an be recast as
\begin{align*}
\nabla_\mu \s -\ell_\mu=0, \qquad \nabla_\mu \ell_\nu -P_{\mu\nu} \s + g_{\mu\nu} \rho=0.
\end{align*}
One only has to find a constraint equation on $\rho$ to close the system. This is done by applying $\nabla$ on the second equation above and after some algebra, so that finally the second-order differential AE equation \eqref{AE1} is replaced by the linear system
\begin{align}
\label{AE2}
\nabla_\mu \s -\ell_\mu=0, \qquad \nabla_\mu \ell_\nu -P_{\mu\nu} \s + g_{\mu\nu} \rho=0, \qquad \nabla_\mu \rho + g^{\alpha\beta}P_{\mu \alpha} \ell_\beta=0.
\end{align}
This system can be rewritten as the action of a linear operator $\nabla_\mu^\T$ acting on the triplet $t=(\s, \ell_\nu, \rho) \in \doubleR^{n+2}$:
\begin{align}
\label{AE3}
\nabla_\mu^\T t =0, \qquad \Rightarrow \qquad \d_\mu\begin{pmatrix}[1.2]  \s \\ \ell_\nu \\ \rho \end{pmatrix} +  \begin{pmatrix}[1.2]  0 & -\delta^\alpha_\mu & 0 \\ -P_{\mu\nu} &  -{\Gamma^\alpha}_{\mu\nu} &  g_{\mu\nu} \\0 & g^{\alpha \beta}P_{\mu\beta} & 0 \end{pmatrix}\begin{pmatrix}[1.2]  \s \\ \ell_\alpha \\ \rho \end{pmatrix}=0.
\end{align}
Given the \emph{particular} definition of $(\s, \ell_\mu, \rho)$, under a Weyl rescaling of the metric one finds after some algebra and using the well known relations $\h{\Gamma^\alpha}_{\mu\nu}= {\Gamma^\alpha}_{\mu\nu} + \delta^\alpha_\mu \Upsilon_\nu + \delta^\alpha_\nu \Upsilon_\mu - g^{\alpha\beta}\Upsilon_\beta g_{\mu\nu}$ and $\h P_{\mu\nu}= P_{\mu\nu} + \nabla_\mu\Upsilon_\nu  - \Upsilon_\mu\Upsilon_\nu + \tfrac{1}{2} \Upsilon^2 g_{\mu\nu}$, where  $\Upsilon_\mu:=z\-\d_\mu z$ and $\Upsilon^2=g^{\alpha\beta}\Upsilon_\alpha\Upsilon_\beta$,
\hspace{-2mm}
\begin{align}
\h \s&=z\s, \notag\\[-6mm]
\h\ell_\mu&
=z\left(\ell_\mu + \Upsilon_\mu \s\right), \qquad \qquad \qquad   \text{ Or in matrix form, }\qquad  \begin{pmatrix}[1.2]  \h\s \\ \h\ell_\mu \\ \h\rho \end{pmatrix} =  \begin{pmatrix}[1.2] z & 0 & 0 \\ z\Upsilon_\mu & z \1 & 0 \\ -z^{-1}\tfrac{1}{2}\Upsilon^2 & -z^{-1} g^{\nu\mu} \Upsilon_\nu & z^{-1}\end{pmatrix}\begin{pmatrix}[1.2]  \s \\ \ell_\mu \\ \rho \end{pmatrix}. \label{GTtractor}\\[-6mm]
\h\rho&
=z^{-1} \left( \rho  - g^{\nu\mu} \Upsilon_\nu \ell_\mu - \tfrac{1}{2}\Upsilon^2 \s \right). \notag
\end{align}
 This, one may consider as a gauge transformation so that the \emph{generic} triplets $t=(\s, \ell_\mu, \rho)$ gauge-related by \eqref{GTtractor}, called \emph{tractors}, are considered as sections (or equivariant maps) of a vector bundle over $(\M, c)$ with fiber $\doubleR^{n+2}$: the so-called \emph{standard tractor bundle} $\T$. 
 
With still more algebra, one shows that this gauge-equivalence still holds for the triplet defined by \eqref{AE2},
\begin{align}
\label{GTtractor-connection}
\begin{pmatrix}[1.2]  \h{(\nabla_\mu \s -\ell_\mu)} \\ \h{(\nabla_\mu \ell_\nu -P_{\mu\nu} \s + g_{\mu\nu} \rho)} \\ \h{(\nabla_\mu \rho + g^{\alpha\beta}P_{\mu \alpha} \ell_\beta)} \end{pmatrix} =  \begin{pmatrix}[1.2] z & 0 & 0 \\ z\Upsilon_\nu & z \1 & 0 \\ -z^{-1}\tfrac{1}{2}\Upsilon^2 & -z^{-1} g^{\nu\alpha} \Upsilon_\alpha & z^{-1}\end{pmatrix}\begin{pmatrix}[1.2]   \nabla_\mu \s -\ell_\mu \\ \nabla_\mu \ell_\nu -P_{\mu\nu} \s + g_{\mu\nu} \rho \\ \nabla_\mu \rho + g^{\alpha\beta}P_{\mu \alpha} \ell_\beta \end{pmatrix}.
\end{align}
So the linear operator $\nabla_\mu^\T$  \eqref{AE3} defines a covariant derivative on $\T$ usually called the \emph{tractor connection}. A tractor satisfying $\nabla_\mu^\T t=0$ is said parallel. By construction, parallel tractors are in bijective correspondence with solutions of the AE equation. 
There is a well defined bilinear form on sections $t, t'\in \Gamma(\T)$ defined by 
\begin{align}
\label{metric_tractor}
\langle t, t'  \rangle=\rho\s' + \ell_\mu g^{\mu\nu} \ell'_\nu +\s \rho' = (\s, \ell_\mu, \rho) \begin{pmatrix} 0 & 0 & 1 \\ 0 & g^{\mu\nu}& 0 \\ 1 & 0 & 0 \end{pmatrix}\begin{pmatrix} \s' \\ \ell'_\nu \\ \rho' \end{pmatrix}=t^TGt', 
\end{align}
where $G$ is a $(r+1, s+1)$-metric on $\T$. Indeed it is invariant under Weyl rescaling $\langle \h t, \h {t'}\rangle=\langle t, t' \rangle$, as can be verified via \eqref{GTtractor}. One also checks via \eqref{AE3} that, like a Levi-Civita connection, the tractor connection preserves the metric thus defined since $\nabla_\T \langle t, t' \rangle = 2 \langle \nabla^\T t, t' \rangle$. 

The commutator of the tractor connection defines the \emph{tractor curvature}
\begin{align}
\label{tractor_curvature}
\left[\nabla_\mu^\T, \nabla_\lambda^\T\right]t=\Omega_{\mu\lambda} t=\begin{pmatrix}[1.2]  0 & 0 & 0 \\ -C_{\mu\lambda, \nu} &  {W^\alpha}_{\mu\lambda, \nu} &  0 \\0 & g^{\alpha \beta}C_{\mu\lambda, \beta} & 0 \end{pmatrix}\begin{pmatrix}[1.2]  \s \\ \ell_\alpha \\ \rho \end{pmatrix},
\end{align}
where $C_{\mu\lambda, \nu}=\nabla_\lambda P_{\mu\nu}$ is the Cotton tensor, and ${W^\alpha}_{\mu\lambda, \nu}$ is the Weyl tensor. 
From this one sees immediately that the tractor connection $\nabla_\mu^\T$ is flat if and only if $(\M, c)$ is conformally flat. 

We refer the reader to \cite{Bailey-et-al94, Curry-Gover2015} for the detailed calculations and further important considerations about tractors and their applications.
\medskip 

Thus is constructed the tractor bundle $\T$ endowed with the tractor connection $\nabla^\T$, bottom up from the AE equation on a conformal manifold $(\M, c)$.  The tractor calculus then provided is thought of as the analog for conformal manifolds of the Ricci tensorial calculus for Riemannian manifolds $(\M, g)$. This approach, while presenting the advantage of being explicit, involves a fair amount of computation in order to derive the basic objects and their transformation properties. In the next section we lay our case that these very objects can be recovered with much less computation, top-down from the conformal Cartan bundle and its Cartan connection via the dressing field method. By doing so, the nature of the tractors and tractor connection as gauge fields of the non-standard kind described in section \ref{The composite fields as  a new kind of gauge fields}  is made clear. 

\subsection{Top-down gauge theoretic approach via the Cartan bundle} 
\label{Top-down gauge theoretic approach via the Cartan bundle} 

The description of the conformal Cartan geometry requires some defining and comments. Once this is done in the following subsection, the dressing field method is applied in the next.

\subsubsection{The Cartan bundle and its naturally associated vector bundle}  
\label{The Cartan bundle and its naturally associated vector bundle}  

The conformal Cartan geometry $(\P, \varpi)$ is said modeled on the Klein model $(G, H)$ where $G=PSO(r+1, s+1)=\left\{ M \in GL_{n+2} | M^T \Sigma M= \Sigma, \det{M}=1 \right\}/ \pm \id$ with $\Sigma=\begin{psmallmatrix}  0 & 0 & -1 \\ 0 & \eta & 0 \\ -1 & 0 & 0 \end{psmallmatrix}$, $\eta$ the flat metric of signature $(r, s)$, and $H$ is a parabolic subgroup such that the Homogeneous space $G/H \simeq (S^r \times S^s) / \mathbb{Z}^2$ is the conformal compactification of what we call with slight abuse Minkowski space, $(\doubleR^n, \eta)$. The structure group of the conformal Cartan bundle $\P(\M, H)$ comprises Lorentz, Weyl and conformal boost symmetries and is described as \cite{Cap-Slovak09, Sharpe}
\begin{align*}
 H = K_0\, K_1=\left\{ \begin{pmatrix} z &  0 & 0  \\  0  & S & 0 \\ 0 & 0 & z^{-1}  \end{pmatrix}\!  \begin{pmatrix} 1 & r & \tfrac{1}{2}rr^t \\ 0 & \1 & r^t \\  0 & 0 & 1\end{pmatrix}  \bigg|\ z\in W:=\doubleR^*_+,\ S\in SO(r, s), 
\ r\in \doubleR^{m*} \right\}.
\end{align*} 
Here ${}^t$ stands for the $\eta$-transposition, namely for the row vector $r$ one has $r^t = (r \eta^{-1})^T$ (the operation ${}^T\,$ being the usual matrix transposition), and $\doubleR^{m*}$ is the dual of $\doubleR^m$.  
Clearly $K_0\simeq CO(r, s)$ via $(S, z) \rarrow zS$, and $K_1$ is the abelian group of conformal boosts. 
The corresponding Lie algebras  $(\LieG, \LieH)$ are graded \cite{Kobayashi}:  $[\LieG_i, \LieG_j] \subseteq \LieG_{i+j}$, $i,j=0,\pm 1$ with the abelian Lie subalgebras $[\LieG_{-1}, \LieG_{-1}] = 0 = [\LieG_1, \LieG_1]$. They decompose respectively as, $\LieG=\LieG_{-1}\oplus\LieG_0\oplus\LieG_1 \simeq \doubleR^m\oplus\co(r,s)\oplus\doubleR^{m*}$ and $\LieH=\LieG_0\oplus\LieG_1 \simeq \co(r,s)\oplus\doubleR^{m*}$. In matrix notation we have,
\begin{align*}
\mathfrak{g} = \left\{ 
\begin{pmatrix} \epsilon &  \iota & 0  \\  \tau  & v & \iota^t \\ 0 & \tau^t & -\epsilon  \end{pmatrix} \bigg|\ (v-\epsilon\1)\in \mathfrak{co}(r, s),\ \tau\in\mathbb{R}^m,\ \iota\in\mathbb{R}^{m*}  
\right\} 
\supset
\LieH = \left\{ \begin{pmatrix} \epsilon &  \iota & 0  \\  0  & v & \iota^t \\ 0 & 0 & -\epsilon  \end{pmatrix} \right\},
\end{align*} 
with the $\eta$-transposition $\tau^t = (\eta\tau)^T$ of the  column vector $\tau$.
The graded structure of the Lie algebras is automatically handled by the matrix commutator.

The Cartan bundle $\P$ is then endowed with the conformal Cartan connection, whose local representative on $\U \subset \M$ is $\varpi  \in \Lambda^1(\U , \LieG)$ with curvature $\b\Omega\in \Lambda^2(\U, \LieG)$. They have the matrix representation
\begin{align*}
\varpi =\begin{pmatrix} a & P & 0 \\ \theta & A & P^t \\0 & \theta^t & -a \end{pmatrix}, \qquad \text{and} \quad  \b\Omega=d\varpi+\varpi^2=\begin{pmatrix} f & C & 0 \\ \Theta & W & C^t \\0 & \Theta^t & -f \end{pmatrix}.
\end{align*}
The soldering part of $\varpi$ is $\theta=e\cdot dx$, i.e with indices $\theta^a:={e^a}_\mu dx^\mu$, with $e={e^a}_\mu$ the so-called vielbein or tetrad field.\footnote{ Notice that from now on we shall make use of  ``$\cdot$'' to denote Greek indices contractions, while Latin indices contraction is naturally understood from matrix multiplication. }
A metric $g$ of signature $(r, s)$ on $\M$ is induced from $\eta$ via  $\varpi$ according to $g(X, Y):=\eta\left( \theta(X), \theta(Y)\right)=\theta(X)^T\eta \theta(Y)$, or in a way more familiar to physicists $g:=e^T\eta e \rarrow g_{\mu\nu}={e_\mu}^a\eta_{ab} {e^b}_\nu$.

It should be noted that the gauge structure $(\P, \varpi)$ on $\M$ is not equivalent to a conformal class of metrics $c$ on it. As we show soon, the action of the local gauge group $\H_\text{\tiny{loc}}$ on $\varpi$ indeed induces a conformal class of metrics via its soldering part, but the degrees of freedom of $\varpi$ compensated for by the gauge symmetry $\H_\text{\tiny{loc}}$ still amounts to more than $\sfrac{n(n+1)}{2} -1=[c]$. 

But there is a way to make the Cartan geometry equivalent to a conformal manifold $(\M, c)$. In a way similar to the singling out of the Levi-Civita connection among all linear connections as the unique torsion-free and metric compatible connection, one can single out the so-called \emph{normal} conformal Cartan connection $\varpi_\text{\tiny{N}}$ as the unique one satisfying the constraints 
$\Theta=0$ (torsion free) and  ${W^a}_{bad}=0$.
 Together with the $\LieG_{-1}$-sector of the Bianchi identity $d\b\Omega+[\varpi, \b\Omega]=0$, these constraints imply $f=0$ (trace free), so that the curvature of the normal Cartan connection reduces to 
 $\b\Omega_\text{\tiny{N}}=\begin{psmallmatrix} 0 & C & 0  \\ 0 & W & C^t \\[0.5mm] 0 & 0 & 0 \end{psmallmatrix}$.
 From the normality condition ${W^a}_{bad}=0$  follows that $P$ 
 has components (in the $\theta$ basis of $\Omega^\bullet (\U)$)
 $P_{ab}=-\frac{1}{(n-2)} \left( R_{ab} - \frac{R}{2(n-1)}\eta_{ab} \right)$,
 where $R$ and $R_{ab}$ are the Ricci scalar and Ricci tensor associated with the $2$-form $R=dA+A^2$. In turn, from this  follows that $W= R + \theta P + P^t\theta^t$ is the well known Weyl $2$-form. By the way, in the gauge $a=0$, $C:=dP +P A=DP$ looks like the familiar Cotton $2$-form.

 The gauge structure $(\P, \varpi_\text{\tiny{N}})$ is  indeed equivalent to a conformal class of metric $c$ on $\M$. However, it would be hasty to then identify $A$ in $\varpi$ or $\varpi_\text{\tiny{N}}$ with the spin connection one is familiar with in physics, and by a way of consequence to take $R:=dA+A^2$ and $P$ as the Riemann and Schouten tensors. Indeed, contrary to expectations $A$ is invariant under Weyl rescaling and neither $R$ nor $P$ have the known Weyl transformations, see \eqref{GT_0} below. It turns out that one recovers the spin connection and the mentioned associated tensors only after  a dressing operation. See the next subsection. 
\medskip

The defining representation space for $G$ is $\doubleR^{n+2}$. It is obviously also a representation for $H$ so that one may form the vector bundle $E=\P \times_H \doubleR^{n+2}$ naturally associated to the Cartan bundle $\P(\M, H)$. Sections of $E$ are $H$-equivariant maps on $\P$ whose local expression is 
\begin{align*}
\vphi: \U \subset \M \rarrow \doubleR^{n+2}, \quad \text{given explicitely as column vectors }\quad \vphi=\begin{pmatrix} \rho\\[1mm] \ell \\ \s  \end{pmatrix}, \quad \text{ with }  \ell=\ell^a \in \doubleR^n, \text{ and } \rho, \s \in \doubleR.
\end{align*}
The covariant derivative induced by the Cartan connection is $D\vphi=d\vphi+\varpi\vphi$. The group metric $\Sigma$ naturally defines an invariant bilinear form on sections of $E$: given $\vphi, \vphi' \in \Gamma(E)$ one has
\begin{align*}
\langle \vphi, \vphi'  \rangle= \vphi^T \Sigma \vphi'=(\rho, \ell^T, \s) \begin{pmatrix}   0 & 0 & -1 \\ 0 & \eta & 0 \\ -1 & 0 & 0  \end{pmatrix} \begin{pmatrix} \rho' \\[1mm] \ell' \\ \s' \end{pmatrix}= -\s\rho' + \ell^T \eta \ell' - \rho\s'.
\end{align*} 
The covariant derivative $D$ naturally preserves this bilinear form since $\varpi$ is $\LieG$-valued: $D\Sigma=d\Sigma + \varpi^T \Sigma + \Sigma \varpi=0$.

$E$ would be called the \emph{standard tractor bundle} in the general terminology of \cite{Cap-Slovak09}. However its sections and covariant derivative thereof do not undergo  the defining Weyl transformation of a tractor as defined in section \ref{Bottom-up construction via prolongation of the Almost Einstein equation}.
Indeed an element $\gamma$ of the local gauge group $\H=\K_0\K_1$ (we now drop the subscript ``loc'') can be factorized as $\gamma=\gamma_0\gamma_1: \U \rarrow H=K_0\, K_1$ with $\gamma_0\in \K_0 := \left\{ \gamma :\U\rarrow K_0  \right\}$ and $\gamma_1 \in \K_1:= \left\{  \gamma :\U\rarrow K_1 \right\}$.
Accordingly, through simple matrix calculations, the gauge transformations of $\vphi$ w.r.t $\K_0$ and $\K_1$ are found to be
\begin{align}
\label{GT_vphi}
\vphi^{\gamma_0}= {\gamma_0}\-\vphi \quad \rarrow \quad \begin{pmatrix} \rho^{\gamma_0} \\[1mm] \ell^{\gamma_0} \\ \s^{\gamma_0}  \end{pmatrix}=\begin{pmatrix} z\- \rho  \\[1mm] S\- \ell \\ z\s  \end{pmatrix}, \qquad \text{ and } \qquad 
\vphi^{\gamma_1}= {\gamma_1}\-\vphi \quad \rarrow \quad \begin{pmatrix} \rho^{\gamma_1} \\[1mm] \ell^{\gamma_1} \\ \s^{\gamma_1}  \end{pmatrix}=\begin{pmatrix} \rho -r\ell + \tfrac{\s}{2}rr^t \\[1mm] \ell^a -r^t\s \\ \s  \end{pmatrix}. 
\end{align}
 The same goes for $D\vphi^{\gamma_0}$ and $D\vphi^{\gamma_1}$. In the first relation put $S=\1$, compare with \eqref{GTtractor} and notice the difference. It is clear that as it stands, $E$ is not the standard tractor bundle $\T$ as previously defined. 
As for the Cartan connection, its  gauge transformation  w.r.t $\K_0$ is
\begin{align}
\label{GT_0}
\varpi^{\gamma_0}& =\gamma_0\-\varpi\gamma_0 + \gamma_0\-d\gamma_0 , \\[2mm]
\begin{pmatrix}
a^{\gamma_0} & P^{\gamma_0} & 0 \\ {\theta^{}}^{\gamma_0} & {A^{}}^{\gamma_0} & (P^{\gamma_0})^t \\  0 & ({\theta^{}}^{\gamma_0})^t & - a^{\gamma_0} 
\end{pmatrix}
&= \begin{pmatrix} a+ z\-dz & z\-P S &  0 \\ S\-\theta z & S\-AS +S\-dS & S\-P^t z\-  \\  0  & z\theta^t S &  -a +zdz\-  \end{pmatrix},  \notag
\end{align}
 and w.r.t $\K_1$ it reads
\begin{align}
\label{GT_1}
\varpi^{\gamma_1}&= \gamma_1\-\varpi\gamma_1 + \gamma_1\-d\gamma_1 ,\\[2mm]
\begin{pmatrix}
a^{\gamma_1} & P^{\gamma_1} & 0 \\ {\theta^{}}^{\gamma_1} & {A^{}}^{\gamma_1} & (P^{\gamma_1})^t \\  0 & ({\theta^{}}^{\gamma_1})^t & -a^{\gamma_1} 
\end{pmatrix}
&= \begin{pmatrix} a-r\theta &\ ar - r\theta r +P -rA +\frac{1}{2}rr^t\theta^t +dr &  0 \\ \theta  & \theta r + A - r^t\theta^t &\ \theta \frac{1}{2}rr^t + Ar^t -r^t\theta^t r^t + P^t + r^t a + dr^t\\  0  & \theta^t  &  \theta^t r^t-a  \end{pmatrix}. \notag
\end{align}
It is clear from the transformation of the soldering part, that the metric induced by $\varpi^{\gamma_0}$ is $z^2g$. Thus the action of $\H$ on $\varpi$ induces a conformal class of metric $c$ on $\M$. 
\medskip

Now that we have the necessary familiarity with the conformal Cartan bundle, its Cartan connection and its naturally associated vector bundle, we are ready to apply the dressing field approach.

\subsubsection{Tractors from gauge symmetry reduction via dressing}  
\label{Tractors from gauge symmetry reduction via dressing}  

A detailed  analysis of the dressing field method applied to the conformal Cartan bundle has been given in \cite{FLM2015_II}. For the benefit of the reader we reproduce here relevant pieces of information, but in a more clear and systematic way based on section \ref{Reduction of gauge symmetries: the dressing field method}. In doing so we also correct few misprints in the results of the mentioned paper. 

Given the decomposition $H=K_0K_1$ we first aim at erasing the conformal boost gauge symmetry $\K_1$ through a dressing field. The most natural choice would be 
\begin{align*}
 u_1 :\U\rarrow K_1, \qquad \text{ that is } \quad u_1=\begin{pmatrix}  1 & q & \tfrac{1}{2}qq^t \\ 0 & \1 & q^t \\ 0 & 0 & 1 \end{pmatrix}.
\end{align*}
No such field jumps out, but it turns out that we may find one via the ``gauge-like'' constraint that requires that the trace of the $CO(r, s)$ part of the composite field $\varpi^{u_1}$ vanishes, explicitly: $\Sigma(\varpi^{u_1}):=\Tr(A^{u_1} - a^{u_1})=-na^{u_1}=0$. This gives the equation $a-q\theta=0$, which once solved for $q$ gives $q_a=a_\mu {e^\mu}_a$, or in index free notation $q=a\cdot e\-$.\footnote{Beware of the fact that in this index free notation $a$ is the set of components of the $1$-form $a$. This should be clear from the context.}
 
 Using \eqref{GT_1} one finds that $q^{\gamma_1}=a^{\gamma_1}\cdot (e^{\gamma_1})\-=(a-re)\cdot e\-=q-r$. This is an abelian dressing transformation which as two consequences. First one checks easily that the constraint $\Sigma(\varpi^{u_1})=0$ is $\K_1$-invariant. From our general discussion in section \ref{Local aspects and Physics} it follows that $u_1$ is a dressing field. And indeed, from $q^{\gamma_1}=q-r$ we find that 
 \begin{align*}
    \begin{pmatrix}  1 & q^{\gamma_1} & \tfrac{1}{2}q^{\gamma_1}{q^{\gamma_1}}^t \\ 0 & \1 & {q^{\gamma_1}}^t \\ 0 & 0 & 1 \end{pmatrix}= \begin{pmatrix}  1 & -r & \tfrac{1}{2}rr^t \\ 0 & \1 & -r^t \\ 0 & 0 & 1 \end{pmatrix}\begin{pmatrix}  1 & q & \tfrac{1}{2}qq^t \\ 0 & \1 & q^t \\ 0 & 0 & 1 \end{pmatrix},
    \quad \text{ that is indeed } \quad u_1^{\gamma_1}=\gamma_1\- u_1.
 \end{align*}
With this $\K_1$-dressing field we can apply - the local version of - proposition \ref{P1} and form the $\K_1$-invariant composite fields
\begin{align}
\label{CompFields_1}
\varpi_1:&=\varpi^{u_1}=u_1 \- \varpi u_1 + u_1\-du_1=\begin{pmatrix} 0 & P_1 & 0 \\ \theta & A_1 & P_1^t \\0 & \theta^t & 0 \end{pmatrix}, \qquad \b\Omega_1:=\b\Omega^{u_1}=u_1\-\b\Omega u_1=d\varpi_1+\varpi_1^2=\begin{pmatrix} f_1 & C_1 & 0 \\ \Theta & W_1 & C_1^t \\0 & \Theta^t & -f_1 \end{pmatrix}, \notag \\[1mm]
\vphi_1:&=u_1\-\vphi=\begin{pmatrix} \rho_1\\[1mm] \ell_1 \\ \s  \end{pmatrix}, \qquad \text{and}\quad D_1\vphi_1=d\vphi_1+\varpi_1 \vphi_1=\begin{pmatrix} d\rho_1 +P_1 \ell_1\\[1mm] d\ell_1+A_1\ell_1 + \theta \rho_1 + P_1^t \s \\ d\s+\theta^t \ell_1  \end{pmatrix}=\begin{pmatrix} \nabla\rho_1 +P_1 \ell_1\\ \nabla\ell_1 + \theta \rho_1 + P_1^t \s \\ \nabla\s+\theta^t \ell_1  \end{pmatrix}
\end{align}
As is usual ${D_1}^2 \vphi_1 = \b\Omega_1 \vphi_1$. We  notice that $f_1=P_1\w\theta$ is the antisymmetric part of the tensor $P_1$.  

The claim is twofold. First, we assert that $\vphi_1$ is a tractor and that the covariant derivative $D_1$ induced from the dressed Cartan connection $\varpi_1$ is a ``generalized'' tractor connection, both written in an orthonormal basis (latin indices). Second, the composite fields \eqref{CompFields_1} are gauge fields of a non-standard kind - such as described in section \ref{The composite fields as  a new kind of gauge fields} - w.r.t Weyl symmetry, but genuine gauge fields - according to section \ref {The composite fields as genuine gauge fields} - w.r.t Lorentz symmetry. Both assertions are supported by the analysis of the residual gauge transformations of these composite fields. 

\paragraph{Residual gauge symmetries} Being by construction $\K_1$-invariant, the composite fields \eqref{CompFields_1} are expected to display a $\K_0\simeq CO(r, s)$-residual gauge symmetry. This group breaks down as a direct product of the Lorentz and Weyl group, $K_0=\mathsf{SO}(r, s) \times \mathsf W$. We then focus on Weyl symmetry first, then only bring our attention to Lorentz symmetry. 
\medskip

The residual transformation of the composite fields under the Weyl gauge group $\W:=\left\{ Z:\U \rarrow \mathsf W\ |\  Z^{Z'}=Z  \right\}$, $Z=\gamma_{0|S=\1}$, are inherited from that of the dressing field $u_1$. Using \eqref{GT_0} to compute $q^Z=a^Z\cdot (e^Z)\-$, one easily finds that 
\begin{align}
\label{WeylGT_u_1}
&u_1^Z=Z\- u_1 C(z)\quad  \text{ where the map }\quad  C: W \rarrow K_1 \mathsf W \subset H \quad \text{ is defined by }\\[1mm]
&C(z):=k_1(z)Z =\begin{pmatrix} 1 & \Upsilon \cdot e\- & \tfrac{1}{2} \Upsilon ^2 \\ 0 & \1 & (\Upsilon\cdot e\-)^t \\ 0 & 0 & 1 \end{pmatrix} \begin{pmatrix} z & 0 & 0 \\ 0 & \1 & 0 \\ 0 & 0 & z\- \end{pmatrix}= \begin{pmatrix}[1.2] z & \Upsilon \cdot e\- &  \tfrac{z\-}{2} \Upsilon ^2 \\ 0 & \1 & z\-(\Upsilon\cdot e\-)^t \\ 0 & 0 & z\-  \end{pmatrix}=\begin{pmatrix}[1.2] z & \Upsilon_a & \tfrac{z\-}{2} \Upsilon ^2 \\ 0 & \1 & z\- \eta^{ab}\Upsilon_b \\ 0 & 0 & z\-  \end{pmatrix}. \notag
\end{align}
 To make the notation explicit $\Upsilon=\Upsilon_\mu=z\-\d_\mu z$, so $\Upsilon\cdot e\- = \Upsilon_\mu {e^\mu}_a=: \Upsilon_a$, and $\Upsilon^2=\Upsilon_a \eta^{ab} \Upsilon_b$. Elements of type $C(z)$ have been called ``tractor gauge transformations'' in \cite{Gover-Shaukat-Waldron09, Gover-Shaukat-Waldron09-2}, which may be deemed inaccurate since contrary to  elements of a genuine gauge group, they do not form a group: $C(z)C(z')\neq C(zz')$. 
 
 Actually \eqref{WeylGT_u_1} is a local instance of Proposition \ref{P4} with $C$ a  $1$-$\alpha$-cocycle satisfying Proposition \ref{alpha_cocycle}. Indeed one can check that  $C(zz')=C(z'z)=C(z')\ {Z'}\- C(z) Z'$, which is the defining property of an abelian $1$-$\alpha$-cocycle. Furthermore, under a further $\W$-gauge transformation and due to $e^Z=ze$, one has $k_1(z)^{Z'}= {Z'}\- k_1(z)Z'$, which implies $C(z)^{Z'}= {Z'}\- C(z) Z'$. So if $u_1$ undergoes a a further $\W$-gauge transformation we have
 \begin{align*}
 \left(u_1^Z\right)^{Z'}=\left(Z^{Z'}\right)\- u_1^{Z'} C(z)^{Z'}=Z\-\ {Z'}\- u_1C(z')\ {Z'}\- C(z) Z'= (ZZ')\- u_1 C(zz'). 
 \end{align*}
 All this implies that the composite fields \eqref{CompFields_1} are indeed instances of gauge fields of the new kind described in section \ref{The composite fields as  a new kind of gauge fields}. As a consequence, by Proposition \ref{P5} we have that their residual $\W$-gauge transformations  are
\begin{align}
\varpi_1^Z&=C(z)\- \varpi_1 C(z) + C(z)\-dC(z)= \begin{pmatrix}[1.2] 0 & z\-\big( P_1 + \nabla(\Upsilon\!\cdot\! e\-) - (\Upsilon\!\cdot\! e\-) \theta (\Upsilon\!\cdot \!e\-)+ \tfrac{1}{2} \Upsilon^2 \theta^t \big)  &  0  \\  z\theta &  A_1 + \theta(\Upsilon\!\cdot \!e\-) - (\Upsilon\!\cdot \!e\-)^t\theta^t & * \\0 & z\theta^t & 0 \end{pmatrix},  \label{varpi_1_Z}    \\[1mm]                     
\b\Omega_1^Z&=C(z)\- \b\Omega_1 C(z)=\begin{pmatrix}[1.2]  f_1 \!-\! (\Upsilon\!\cdot\! e\-)\Theta  &   z\-\big( C_1 \!-\! (\Upsilon\!\cdot\! e\-) (W_1 \!-\! f_1) \!-\! (\Upsilon\!\cdot\! e\-)\Theta(\Upsilon\!\cdot\! e\-) \!+\! \tfrac{1}{2} \Upsilon^2\Theta^t  \big) & 0  \\  z\Theta & W_1 \!+\!\Theta (\Upsilon\!\cdot\! e\-) \!-\! (\Upsilon\!\cdot\! e\-)^t\Theta^t & * \\ 0 & z\Theta^t & * \end{pmatrix},           \label{Omega_1_Z}  \\[1mm]                    
\vphi_1^Z&=C(z)\- \vphi_1 = \begin{pmatrix} z\-\left( \rho_1 -(\Upsilon\cdot e\-)\ell_1 + \tfrac{\s}{2} \Upsilon^2 \right) \\[1mm] \ell_1 - (\Upsilon\cdot e\-)^t \s \\ z\s \end{pmatrix}, \qquad \text {and }\quad (D_1\vphi_1)^Z=D_1^Z\vphi_1^Z= C(z)\- D_1\vphi_1.        \label{Tractor_TConnection_1}
\end{align}

Several elements should be highlighted. First, in \eqref{varpi_1_Z} notice that now the Lorentz part $A_1$ of the composite field $\varpi_1$ indeed exhibits the known Weyl transformation for the spin connection, and  $P_1$ transforms as the genuine Schouten tensor in an orthonormal basis. But actually they reduce to these only when one restricts to the dressing of the normal Cartan connection, in which case 
\begin{align*}
\b\Omega_{\text{\tiny N},1}=d\varpi_{\text{\tiny N},1}+\varpi_{\text{\tiny N},1}^2=\begin{pmatrix} 0 & C_1 & 0 \\ 0 & W_1 & C_1^t \\0 & 0 & 0 \end{pmatrix}, \quad 
\text{ and one has } \quad \b\Omega_{\text{\tiny N},1}^Z=C(z)\- \b\Omega_{\text{\tiny N},1}C(z)=\begin{pmatrix} 0  &   z\-\left(C_1 - (\Upsilon\cdot e\-) W_1\right)  & 0  \\ 0 & W_1 & * \\ 0 & 0 & * \end{pmatrix}.
\end{align*}
We then see that $C_1=\nabla P_1$ transforms as - and therefore is - the Cotton tensor, while $W_1$ is the invariant Weyl tensor. 

Then, most importantly for our concern, the first relation in \eqref{Tractor_TConnection_1} is the vielbein version of \eqref{GTtractor} so that the dressed section $\vphi_1$ is indeed a tractor field, section of a $C$-vector bundle that we denote $E_1\!=E^{u_1}\!=\P \times_{C( W)}\doubleR^{n+2}$. 
The invariant bilinear form on $E$ defined by the group metric $\Sigma$  is also defined on $E_1$: $\left\langle \vphi_1, \vphi'_1  \right\rangle= \vphi_1^T \Sigma \vphi'_1$. Indeed since $C(z) \in K_1\mathsf W \subset H$, we have $\left\langle \vphi_1^Z, {\vphi'_1}^Z \right\rangle = \left\langle C(z)\-\vphi_1, C(z)\- \vphi'_1 \right\rangle = \vphi_1^T (C(z)\-)^T \Sigma C(z)\- \vphi_1'=\vphi_1^T\Sigma \vphi'_1 = \left\langle \vphi_1, \vphi_1' \right\rangle$.

What's more, $D_1:=d + \varpi_1$ in \eqref{CompFields_1} is a vielbein formulated generalization of the tractor connection \eqref{GTtractor-connection}. But then the term ``connection'', while not inaccurate, could hide the fact that $\varpi_1$ is no more a geometric connection w.r.t Weyl symmetry. So we shall prefer to call $D_1$ a  generalized tractor \emph{covariant derivative}. 
The standard tractor covariant derivative \eqref{AE3} is recovered by restriction to the dressing of the normal Cartan connection: $D_{\text{\tiny N},1}\vphi_1=d\vphi_1 + \varpi_{\text{\tiny N},1} \vphi_1$. Then $D_{\text{\tiny N},1}^2\vphi_1=\b\Omega_{\text{\tiny N},1} \vphi_1$ recovers \eqref{tractor_curvature}.  We note that $\varpi_1$ being $\LieG$-valued, $D_1\Sigma=0$  and $D_1$ preserves the bilinear form $\langle\ \!, \rangle$. 

In short, by erasing via dressing the $\K_1$-gauge symmetry from the conformal Cartan gauge structure $\left( (\P, \varpi), E\right)$ over $\M$, we have already recovered top-down the tractor bundle and tractor covariant derivative in the orthonormal frame formulation, $(E_1, D_{\text{\tiny N},1})$, as a special case of the $C$-vector bundle endowed with a covariant derivative $(E_1, D_1)$. In this scheme they appear as instances of gauge fields of the new kind described in section \ref{The composite fields as a new kind of gauge fields}. 
\medskip

But our analysis of the residual gauge symmetry is not complete yet, since we need to address the Lorentz residual symmetry. Again, the residual transformation of the composite fields \eqref{CompFields_1} under the Lorentz gauge group $\SO:=\left\{ \mathsf S:\U \rarrow \mathsf{SO}(r, s)\ | {\mathsf S}^{\mathsf S'}={\mathsf S'}\- \mathsf S \mathsf S'  \right\}$, $\mathsf{S}=\gamma_{0|z=1}$, is inherited from that of the dressing field $u_1$. Using \eqref{GT_0} to compute $q^{ \mathsf S}=a^{ \mathsf S}\cdot (e^{ \mathsf S})\-=qS$, one easily finds that $u_1^\mathsf S=\mathsf S\- u_1 \mathsf S$. This is a local instance of Proposition \ref{Prop2}, which then allows to conclude that the composites fields \eqref{CompFields_1} are \emph{genuine} gauge fields w.r.t Lorentz gauge symmetry. Hence, from Corollary \ref{Cor3} follows that their residual $\SO$-gauge transformations are
\begin{align}
\label{CompFields_1_S}
&\varpi_1^\mathsf S={\mathsf S}\-  \varpi_1 \mathsf S  + { \mathsf S}\- d \mathsf S= \begin{pmatrix} 0 & P_1 S &  0 \\ S\-\theta  & S\-A_1S +S\-dS & S\-P^t   \\  0  & \theta^t S &  0 \end{pmatrix}, \qquad \b\Omega_1^\mathsf S={ \mathsf S}\- \b\Omega_1  \mathsf S= \begin{pmatrix} f_1 & C_1 S &  0 \\ S\-\Theta  & S\-WS & S\-C_1^t   \\  0  & \Theta^t S &  -f_1  \end{pmatrix},   \notag \\[1mm]
&\vphi_1^\mathsf S= {\mathsf S}\- \vphi_1= \begin{pmatrix}  \rho_1 \\[1mm] S\- \ell_1  \\ \s \end{pmatrix},\qquad \text{ and } \qquad (D_1\vphi_1)^\mathsf S=D_ 1^\mathsf S \vphi_1^\mathsf S = {\mathsf S}\- D_1\vphi_1.
\end{align}
This is to be fully expected since we concluded that e.g $A_1$, $R_1$, $P_1$ are the usual spin connection, Riemann and Schouten tensors in an orthonormal frame. 
We notice in particular that $\vphi_1$ behaves as a standard section of a $\mathsf{SO}$-associated bundle. We should then refine our notation for the bundle $E_1$ and denote it $E_1=\P\times_{ C(\mathsf W)\ \! \mathsf{SO}} \doubleR^{n+2}$. 

The actions of $\SO$ and $\W$ on the composite fields $\chi_1$ are compatible and commutative. Indeed, we have first that $\sS^\W=\sS$ so that on the one hand: $\left(\chi_1^\SO\right)^\W=\left(\chi_1^\mathsf S\right)^\W=\left( \chi_1^\W \right)^{\mathsf S^\W}=\left( \chi_1^{C(z)}\right)^\mathsf S=\chi_1^{C(z)\mathsf S}$. But then we also have $C(z)^\SO={\mathsf S}\- C(z)\mathsf S$, so on the other hand we get: $\left( \chi_1^\W \right)^\SO=\left( \chi_1^{C(z)}\right)^\SO=\left( \chi_1^\SO\right)^{C(z)^\SO}=\left(\chi_1^\mathsf S \right)^{{\mathsf S}\- C(z)\mathsf S}=\chi_1^{C(z)\mathsf S}$.
\smallskip

So, as the considerations at the end of section \ref{The composite fields as genuine gauge fields} make clear, the fact that the composite fields \eqref{CompFields_1} are genuine $\SO$-gauge fields satisfying \eqref{CompFields_1_S} entitles us to ask if a further dressing operation aiming at erasing Lorentz symmetry is possible. The answer is yes. 

\paragraph{Further dressing operation: erasing Lorentz symmetry and residual Weyl symmetry} In order not to spoil the $\K_1$-invariance obtained from the first dressing $u_1$, a dressing $\u$ for the Lorentz gauge symmetry  should be $\K_1$-invariant. So it would be defined by the relations:  $\u^\mathsf S={\mathsf s}\-\u$ and $\u^{\gamma_1}=\u$. In a more usual language, it is about switching from orthonormal frames to 
holonomic frames, so the vielbein $e={e^a}_\mu$ in the soldering form $\theta=e\cdot dx$ is a natural candidate. And as a matter of fact we see from \eqref{GT_0} and \eqref{GT_1} that $e^S=S\- e$ and $e^{\gamma_1}=e$. Then we can form the $\SO$-dressing field
\begin{align*}
\u=\begin{pmatrix} 1 & 0 & 0 \\ 0 & e & 0 \\ 0 & 0 & 1 \end{pmatrix}\quad  \text{ which satisfies } \quad \begin{pmatrix} 1 & 0 & 0 \\ 0 & e^S & 0 \\ 0 & 0 & 1 \end{pmatrix}=\begin{pmatrix} 1 & 0 & 0 \\ 0 & S\- & 0 \\ 0 & 0 & 1 \end{pmatrix}\begin{pmatrix} 1 & 0 & 0 \\ 0 & e & 0 \\ 0 & 0 & 1 \end{pmatrix}\quad\text{ that is } \quad \u^\mathsf S= {\mathsf S}\- \u, \text{ and } \u^{\gamma_1}=\u.
\end{align*}
With this $\SO$-dressing field we can apply again - the local version of - proposition \ref{P1} on the composite fields $\chi_1$ \eqref{CompFields_1} and form the $\K_1$- and $\SO$-invariant composite fields $\chi_\l:=\chi_1^{\u}$:
\begin{align}
\label{CompFields_L}
\varpi_\l:&=\varpi_1^{\u} = \u\- \varpi_1 \u + \u\-d\u
=\begin{pmatrix} 0 & \sP & 0 \\ dx & \Gamma & g\-\!\cdot\!\sP \\0 & g\!\cdot\! dx & 0 \end{pmatrix}
=\begin{pmatrix} 0 & \sP_{\mu\nu} & 0 \\ \delta^\rho_\mu & {\Gamma^\rho}_{\mu\nu} & g^{\rho\alpha}\sP_{\mu\alpha} \\0 & g_{\mu\nu} & 0 \end{pmatrix} dx^\mu, \\[2mm]
 \b\Omega_\l:&=\b\Omega_1^{\u}=\u\-\b\Omega_1 \u=d\varpi_\l+\varpi_\l^2
 =\begin{pmatrix} f_\l & \sC & 0 \\ \sT & \sW & g\-\!\cdot\!\sC \\0 & g\!\cdot\!\sT & -f_\l \end{pmatrix}
 =\frac{1}{2}\begin{pmatrix}[1.2] \sP_{[\mu\lambda]} & \sC_{\nu, \mu\lambda,} & 0 \\ {\sT^\rho}_{\mu\lambda} & {\sW^\rho}_{\nu, \mu\lambda} & g^{\rho\alpha}\sC_{\alpha, \mu\lambda} \\0 & g_{\nu\alpha}{\sT^{\alpha}}_{\mu\lambda} & -\sP_{[\mu\lambda]} \end{pmatrix} dx^\mu \w dx^\lambda, \notag \\[2mm]
\vphi_\l:&=\u\-\vphi_1=\begin{pmatrix} \rho_\l \\[1mm] \ell_\l \\ \s  \end{pmatrix}, 
\quad \text{and}\quad D_\l\vphi_\l=d\vphi_\l+\varpi_\l \vphi_\l=\begin{pmatrix}[1.1] d\rho_\l +\sP \ell_\l \\ d\ell_\l+ \Gamma \ell_\l + \rho_\l dx + g\-\!\cdot\!\sP \s \\ d\s+g\!\cdot\! dx \ell_\l  \end{pmatrix}
=\begin{pmatrix}[1.1] \nabla\rho_\l +\sP \ell_\l \\ \nabla\ell_\l +  \rho_\l dx + g\-\!\cdot\!\sP \s \\ \nabla\s+g\!\cdot\! dx \ell_\l   \end{pmatrix}. \notag
\end{align}
As is usual ${D_\l}^2 \vphi_\l = \b\Omega_\l \vphi_\l$. 
The claim is now that $\vphi_\l$ is a tractor and that the covariant derivative $D_\l$ induced from the dressed Cartan connection $\varpi_\l$ is a generalized tractor covariant derivative, the usual one being induced by $\varpi_{\n, \l}$. Furthermore the composite fields \eqref{CompFields_L} are  also gauge field of a non-standard kind w.r.t Weyl residual symmetry. This we now substantiate.
\medskip

Notice that as discussed at the end of section \ref{Residual gauge symmetry} the final composite fields $\chi_\l$ could have been obtained from the gauge fields $\chi$ in one step with the $\SO\K_1$-dressing field $u:=u_1\u$. We have indeed on the one hand $u^{\gamma_1}= u_1^{\gamma_1} \u^{\gamma_1}=\gamma_1\-u_1  \u=\gamma_1\- u$, and on the other hand $u^\mathsf S = u_1^\mathsf S\, \u^\mathsf S={\mathsf S}\-u_1 \mathsf S {\mathsf S}\-\u={\mathsf S}\-u$. 

Then the $\W$-residual gauge transformations of the composite fields $\chi_\l$ depend on that of the dressing $u=u_1\u$. We already have \eqref{WeylGT_u_1} for $u_1$, and since $e^Z=ze$ we define  $\u^Z= \t Z \u=\u\t Z$ with $\t Z=\begin{psmallmatrix} 1 & 0 & 0 \\ 0 & z & 0 \\ 0 & 0 & 1  \end{psmallmatrix}$.
From this it is easily found that
\begin{align}
\label{WeylGT_u_L}
&u^Z=Z\- u \b C(z) \quad \text{ where the map } \quad \b C:W \rarrow GL_{n+2} \supset H, \quad \text{ is defined by }\\[1mm]
&\b C(z)= \b k_1(z)  \b Z = \begin{pmatrix}[1.2] 1 & \Upsilon & \tfrac{1}{2} \Upsilon^2 \\ 0 & \1 & g\-\!\cdot\!\Upsilon \\ 0 & 0 & 1 \end{pmatrix} \begin{pmatrix}[1.2] z & 0 & 0 \\ 0 & z & 0 \\ 0 & 0 & z\- \end{pmatrix} =\begin{pmatrix}[1.2] z & z\Upsilon & \tfrac{z\-}{2} \Upsilon^2 \\ 0 & z \1 & z\- g\-\!\cdot\!\Upsilon \\ 0 & 0 & z\- \end{pmatrix}
= \begin{pmatrix}[1.2] z & z\Upsilon_\mu & \tfrac{z\-}{2} \Upsilon^2 \\ 0 & z \1 & z\- g^{\mu\alpha}\Upsilon_\alpha \\ 0 & 0 & z\- \end{pmatrix}.\notag
\end{align}
with $\b k_1(z):={\u}\- k_1(z)\u $ and $\b Z:=Z\t Z$. Also $\Upsilon^2=\Upsilon_\alpha g^{\alpha\beta} \Upsilon_\beta$. 
This is again a local instance of Proposition \ref{P4} with $\b C$ a  $1$-$\alpha$-cocycle satisfying Proposition \ref{alpha_cocycle}. Indeed one can check that  $\b C(zz')=\b C(z'z)=\b C(z')\ {\b Z}^{\prime-1}\b C(z) {\b Z}'$, which is the defining property of an abelian $1$-$\alpha$-cocycle. Furthermore, under a further $\W$-gauge transformation and  one has $\b k_1(z)^{Z'}= {\b Z}^{\prime -1}\b k_1(z) {\b Z}'$, which implies $\b C(z)^{Z'}= {\b Z}^{\prime-1} \b C(z) {\b Z}'$. So if $u$ undergoes a a further $\W$-gauge transformation we have
 \begin{align*}
 \left(u^Z\right)^{Z'}=\left(Z^{Z'}\right)\- u^{Z'} \b C(z)^{Z'}= Z\-\ {Z'}\- u \b C(z')\ {\b Z}^{\prime-1} \b C(z) {\b Z}'= (ZZ')\- u \b C(zz'). 
 \end{align*}
 All this implies that the composite fields \eqref{CompFields_L} are indeed instances of gauge fields of the new kind described in section \ref{The composite fields as  a new kind of gauge fields}. As a consequence, by Proposition \ref{P5} we have that their residual $\W$-gauge transformations  are
\begin{align}
\varpi_\l^Z&=\b C(z)\- \varpi_\l \b C(z) + \b C(z)\-d\b C(z)=
 \begin{pmatrix}[1.2] 0 & \sP + \nabla\Upsilon - \Upsilon \!\cdot\! dx \Upsilon + \tfrac{1}{2} \Upsilon^2 g\!\cdot\!dx & 0 \\ 
dx & \Gamma + z\- dz \delta + \Upsilon dx - g\- \!\!\cdot\! \Upsilon\  g\!\cdot\!dx  \hspace{1mm} & z^{-2} g\- \!\cdot\! \left(*\right) \\ 
0 &\quad z^2 g\!\cdot\!dx & 0 
\end{pmatrix},  \label{varpi_L_Z}    \\[1mm]                     
\b\Omega_\l^Z&=\b C(z)\- \b\Omega_\l \b C(z)= 
\begin{pmatrix}[1.2] f_\l - \Upsilon \!\cdot\! \sT &\ \ \sC - \Upsilon \!\cdot\! \sW + (f_\l - \Upsilon \!\cdot\! \sT)\Upsilon + \tfrac{1}{2} \Upsilon^2  g\!\cdot\!\sT & 0\\
\sT & \sW + \sT\Upsilon - g\- \!\cdot\! \Upsilon\ g\!\cdot\!\sT &\ z^{-2} g\-\!\cdot\!\left(*\right) \\
0 & z^2g\!\cdot\!\sT & *\end{pmatrix} ,           \label{Omega_L_Z}  \\[1mm]                    
\vphi_\l^Z&=\b C(z)\- \vphi_\l = \begin{pmatrix} z\-\left( \rho_\l -\Upsilon\cdot\ell_\l + \tfrac{\s}{2} \Upsilon^2 \right) \\[1mm] z\-\left( \ell_\l - g\-\!\cdot\!\Upsilon \s \right) \\ z\s \end{pmatrix},\qquad \text {and }\quad (D_\l\vphi_\l)^Z=D_\l^Z\vphi_\l^Z= \b C(z)\- D_\l\vphi_\l.        \label{Tractor_TConnection_Z}
\end{align}

For completeness we make some comments paralleling those of the previous situation. First, if in \eqref{varpi_L_Z} $\Gamma$  exhibits the known Weyl transformation for the Levi-Civita connection and  $\sP$ transforms as the genuine Schouten tensor, they actually reduces to these standard objects only if one one restricts to the dressing of the normal Cartan connection, in which case
\begin{align*}
\b\Omega_{\n, \l}=d\varpi_{\n, \l}+\varpi_{\n, \l}^2=\begin{pmatrix} 0 & \sC & 0 \\ 0 & \sW & g\-\!\cdot\!\sC \\0 & 0 & 0 \end{pmatrix}, \quad 
\text{ and one has } \quad \b\Omega_{\n, \l}^Z=\b C(z)\- \b\Omega_{\n, \l}\b C(z)=\begin{pmatrix} 0  &   \sC - \Upsilon\cdot \sW & 0  \\ 0 & \sW & z^{-2} g\- \!\cdot\! \left(*\right) \\ 0 & 0 & 0 \end{pmatrix}.
\end{align*}
We then see that $\sC=\nabla \sP$ is the Cotton tensor while $\sW$ is the invariant Weyl tensor.

Then, most importantly for our concern, the first relation in \eqref{Tractor_TConnection_1} is the $g$-transposed of \eqref{GTtractor}, so that the dressed section $\vphi_\l$ is indeed a tractor, section of a $\b C$-vector bundle that we denote $E_\l\!=E^{u_\l}=\P \times_{\b C( W)} \doubleR^{n+2}$. 

Furthermore, out of the group metric $\Sigma$ one can form the non-invariant metric  $G:=\u^T \Sigma \u= \begin{psmallmatrix} 0 & 0 & -1 \\ 0 & g & 0 \\ -1 & 0 & 0 \end{psmallmatrix}$ which induces an invariant bilinear form on $E_\l$ defined by 
\begin{align*}
\left\langle \vphi_\l, \vphi_\l' \right\rangle_{\text{\tiny{G}}}:= \vphi_\l^T G \vphi_\l' = -\s \rho_\l' + \ell_\l\!\cdot\! g\!\cdot\! \ell_\l' - \rho_\l \s'.
\end{align*}
This is indeed the bilinear form on tractor \eqref{metric_tractor}. Its invariance is most directly proven in our framework. Indeed $G^Z={\t Z}^2G$, so that 
\begin{align*}
{\b C(z)}^{-1T}\ G^Z\  {\b C(z)}\-&= \u^T k_1(z)^{-1T}\u^{-1T}Z\- {\t Z}\- \left( {\t Z}^2 \u^T \Sigma \u \right)\ {\t Z}\- Z\- \u\-k_1(z)\- \u,\\
								&=\u^T \left(k_1(z)Z\right)^{-1T}\  \Sigma\  \left(k_1(z)Z\right)\- \u=\u^T\Sigma \u= G.
\end{align*}
This simple calculation allows to quickly conclude: 
\begin{align*}
\left\langle \vphi_\l^Z, {\vphi_\l'}^Z \right\rangle_{\text{\tiny{$G^Z$}}}:=\left\langle {\b C(z)}\-\vphi_\l, {\b C(z)}\-\vphi_\l' \right\rangle_{\text{\tiny{$G^Z$}}}=\vphi_\l^T{\b C(z)}^{-1T}\ G^Z\  {\b C(z)}\- \vphi_\l' = \vphi_\l^T G \vphi_\l'= \left\langle \vphi_\l, \vphi_\l' \right\rangle_{\text{\tiny{G}}}.
\end{align*}

What's more, the fourth equation in  \eqref{Tractor_TConnection_Z} reproduces \eqref{GTtractor-connection}. So
  $D_\l:=d + \varpi_\l$ in \eqref{CompFields_L} is a generalization of the tractor connection \eqref{AE3}. But   $\varpi_\l$ being no more a geometric connection w.r.t Weyl symmetry due to \eqref{WeylGT_u_L}, instance of Proposition \ref{P4}, we refer to $D_\l$ as a generalized tractor \emph{covariant derivative}. 
The standard tractor covariant derivative is recovered by restriction to the dressing of the normal Cartan connection: $D_{\n, \l}\vphi_\l=d\vphi_\l + \varpi_{\n, \l} \vphi_\l$. Then $D_{\n, \l}^2\vphi_\l\b\Omega_{\n, \l} \vphi_\l$ recovers \eqref{tractor_curvature}. Note that $D_\l G=dG -\varpi_\l^TG- G \varpi_\l=-\u^T\left(  \varpi_1^T \Sigma + \Sigma \varpi_1\right) \u=0$, so $D_\l$ preserve the bilinear form $\left\langle\ \! , \right\rangle_{\text{\tiny{G}}}$. 
\bigskip

To conclude this section, let us sum-up what has been done. By erasing via dressing the $\K_1$ and $\SO$ gauge symmetries from the conformal Cartan gauge structure $\left( (\P, \varpi), E\right)$ over $\M$, we have recovered top-down the tractor bundle and tractor covariant derivative, $(\T, \nabla^\T)=(E_\l, D_{\n, \l})$, as a special case of the $\b C$-vector bundle with covariant derivative $(E_\l, D_\l)$. In the next section we provides the BRST treatment of the two symmetry reductions. 
 
\subsection{BRST treatment}
\label{BRST treatment}

The gauge group of the initial Cartan geometry is $\H$, so the associated ghost $v\in$ Lie$\H$ splits along the grading of $\LieH$,
\begin{align}
v=v_0+v_\iota=v_\epsilon+v_\ss+v_\iota=\begin{psmallmatrix} \epsilon & 0 & 0 \\ 0 & 0 & 0 \\ 0 & 0 & -\epsilon \end{psmallmatrix} + \begin{psmallmatrix} 0 & 0 & 0 \\ 0 & s & 0 \\ 0 & 0 & 0 \end{psmallmatrix}+\begin{psmallmatrix} 0 & \iota & 0 \\ 0 & 0 & \iota^t \\ 0 & 0 & 0 \end{psmallmatrix}.
\end{align} 
The BRST operator splits accordingly as $s=s_0+s_1= s_\ww+s_\l+s_1$. Then the BRST algebra for the gauge fields $\chi=\{ \varpi, \b\Omega, \vphi \}$ is
\begin{align}
s\varpi=-Dv=-dv -[\varpi, v], \qquad s\b\Omega=[\b\Omega, v], \qquad s\vphi=-v\vphi \qquad \text{and} \qquad sv=-v^2,
\end{align}
with the first and third relations in particular reproducing the infinitesimal versions of \eqref{GT_0}, \eqref{GT_1} and \eqref{GT_vphi}. Denote this initial algebra $\mathsf{BRST}$.
As the general discussion of section \ref{Application to the BRST framework} showed, the dressing approach modifies it. 

\paragraph{First dressing} We know from this  general discussion that the composite fields $\chi_1=\{ \varpi_1,\b\Omega_1, \vphi_1 \}$ satisfy a modified BRST algebra formally similar but with composite ghost $v_1:=u_1\- v u_1 + u_1\-su_1$. The inhomogeneous term can be found explicitly from the finite gauge transformations of $u_1$. Writing the linearizations $\gamma_1\simeq\1 + v_\iota$ and $\sS \simeq\1 + v_\ss$, the BRST actions of $\K_1$ and $\SO$ are easily found to be
\begin{align*}
u_1^{\gamma_1}=\gamma_1\-u_1 \quad \rarrow \quad s_1u_1=-v_\iota u_1  \qquad \text{and} \qquad u_1^\sS=\sS\- u_1 \sS \quad \rarrow \quad s_\l u_1=[u_1, v_\ss].
\end{align*}
This shows that the Lorentz sector gives an instance of the general result \eqref{NewGhost2}. Now, defining the linearizations $Z\simeq\1 + v_\epsilon$ and $k_1(z)\simeq\1 + k_1(\epsilon)$, so that $C(z)=k_1(z)Z\simeq\1+c(\epsilon)=\1 + k_1(\epsilon)+v_\epsilon$, the BRST action of $\W$ is
\begin{align*}
u_1^Z=Z\-u_1 C(z) \quad \rarrow \quad s_\ww u_1=-v_\epsilon u_1 + u_1 c(\epsilon).
\end{align*}
This shows that the Weyl sector gives an instance of the general result \eqref{NewGhost3}. We then get the composite ghost
\begin{align}
v_1:&=u_1\- (v_\epsilon+v_\ss+v_\iota)  u_1 + u_1\- (s_\ww+s_\l +s_1) u_1,  \notag \\
	&=u_1\- (v_\epsilon+v_\ss+v_\iota)  u_1 + u_1\-  \big( -v_\epsilon u_1 + u_1 c(\epsilon) \ + \ [u_1, v_\ss] \ - \ v_\iota u_1 \big), \notag\\
	&=c(\epsilon) + v_\ss=\begin{pmatrix} \epsilon & \d\epsilon\!\cdot\!e\- & 0 \\ 0 & s & (\d\epsilon\!\cdot\!e\-)^t \\ 0 & 0 & -\epsilon \end{pmatrix}.
\end{align} 
We see that the ghost of conformal boosts $\iota$ has disappeared from this new ghost. This means that $s_1 \chi_1=0$, which reflects the $\K_1$-gauge invariance of the composite fields $\chi_1$. The composite ghost $v_1$ only depends on $v_\ss$ and $\epsilon$, it encodes the residual $\K_0$-gauge symmetry.
The BRST algebra for the composite fields $\chi_1$ is then explicitly
\begin{align*}
s\varpi_1&=-Dv_1=-dv_1 -[\varpi_1, v_1]
=\begin{pmatrix}[1.2] 0 & - \nabla(\d\epsilon\!\cdot\! e\-) -  P_1(s-\epsilon) & 0 \\
-(s-\epsilon) \theta & -\nabla s- \theta (\d\epsilon\!\cdot\! e\-) - (\d\epsilon\!\cdot\! e\-)^t\theta^t & * \\
0 &\theta^t(s+\epsilon) & 0 \end{pmatrix},
\end{align*}
where $\nabla(\d\epsilon\!\cdot\!e\-)=d(\d\epsilon\!\cdot\! e\-) + (\d\epsilon\!\cdot\! e\-) A_1$, and $\nabla s=ds + [A, s]$, 
\begin{align*}
 &s\b\Omega_1=[\b\Omega_1, v_1]=\begin{pmatrix}[1.2] -(\d\epsilon\!\cdot\!e\-)\Theta &  C_1(s-\epsilon) - (\d\epsilon\!\cdot\! e\-) \left( W_1 -f_1 \right) & 0 \\
-(s-\epsilon) \Theta & [W_1, s] + \Theta(\d\epsilon\!\cdot\!e\-) -  (\d\epsilon\!\cdot\!e\-)^t \Theta^t & * \\ 0 &  \Theta^t(s+\epsilon)  & * \end{pmatrix}, \\[1mm]
&\text{in the normal case } s\b\Omega_{\n, 1}=\begin{pmatrix}[1.2] 0 &  C_1(s\!-\!\epsilon)\!-\!(\d\epsilon\!\cdot\! e\-)  W_1  & 0 \\ 0 & [W_1, s]  & * \\ 0 &  0 & 0 \end{pmatrix},\\[1mm]
&s\vphi_1=-v_1\vphi_1 =\begin{pmatrix} -\epsilon \rho_1- (\d\epsilon\!\cdot\!e\-)\ell_1 \\[1mm]  -s\ell_1 - (\d\epsilon\!\cdot\!e\-)^t\s \\ \epsilon \s \end{pmatrix}=\begin{pmatrix}[1.2] -\epsilon \rho_1- \d_a\epsilon\ \ell_1^a \\[1mm] -{s^a}_b \ell_1^b - \eta^{ab} \d_b\epsilon\ \s \\ \epsilon \s \end{pmatrix}, 
\quad \text{and} \qquad sv_1=-{v_1}^2=\begin{pmatrix} 0 & (\d\epsilon\!\cdot\!e\-) s & 0 \\ 0 & s^2 & * \\ 0 & 0 & 0  \end{pmatrix}.
\end{align*}
Denote this algebra $\mathsf{BRST}_{\ww, \l}$. Since $v_1=c(\epsilon)+v_\ss$, it splits naturally as a Lorentz and a Weyl subalgebras, $s=s_\ww+s_\l$. The Lorentz sector $(s_\l, v_\ss)$, obtained by setting $\epsilon=0$, shows the composites fields $\chi_1$ to be genuine Lorentz gauge fields (compare with \eqref{CompFields_1_S}). While the Weyl sector $(s_\ww, c(\epsilon))$, obtained by setting $s=0$, shows $\chi_1$ to be non-standard Weyl gauge fields (compare with  \eqref{varpi_1_Z}-\eqref{Tractor_TConnection_1}).

\paragraph{Second dressing} Now in the final step, we further modifies $\mathsf{BRST}_{\ww, \l}$ through the dressing $\u$. The composite fields $\chi_\l$ then satisfy a BRST algebra with composite ghost $v_\ww=\u\- v_1 \u + \u\-s\u$. 
Again the inhomogeneous term can be found explicitly from the finite gauge transformations of $u_\l$. Writing the linearization $\t Z\simeq\1 + \t v_\epsilon$ we have
\begin{align*}
\u^\sS=\sS\- \u \quad \rarrow \quad s_\l\u=-v_\ss \u, \qquad \u^Z=\t Z \u \quad \rarrow  \quad s_\ww \u= \t v_\epsilon \u \qquad \text{ and } \qquad \u^{\gamma_1}=\u \quad \rarrow \quad s_\iota \u=0.
\end{align*}
So we get the final composite ghost
\begin{align}
\label{final_ghost1} 
v_\ww&=\u\- (c(\epsilon) + v_\ss) \u + \u\-(s_\ww + s_\l + s_\iota) \u, \notag\\
	&=\u\- (c(\epsilon) + v_\ss) \u + \u\- \left( \t v_\epsilon \u - v_\ss \u \right), \notag\\
	&= \u\- c(\epsilon) \u + \t v_\epsilon= \begin{pmatrix}  \epsilon &  \d\epsilon & 0 \\ 0 & \epsilon \1 & g\-\!\!\cdot\!\d\epsilon \\ 0 & 0 & -\epsilon \end{pmatrix}.
\end{align}

This final composite ghost is also obtained in one step from dressing $v$ with the dressing $u:=\u u_1$ which satisfy
\begin{align*}
u^{\gamma_1}=\gamma_1 u \quad \rarrow \quad s_1 u= -v_\iota u \qquad \text{and} \qquad \qquad u^\sS=\sS\- u \quad \rarrow \quad s_\l u=-v_\ss u.
\end{align*}
Also, defining the linearization $\b Z:=Z \t Z \simeq \1 + \b v_\epsilon=\1 + v_\epsilon+\t v_\epsilon$ and $\b k_1(z):=\u\- k_1(z)\u \simeq \1 + \b k_1(\epsilon)=\1 + \u\- k_1(\epsilon) \u$, we get  $\b C(z)=\b k_1(z)\b Z\simeq\1+\b c(\epsilon)=\1 + \b k_1(\epsilon)+\b v_\epsilon$. Then we have
\begin{align*}
u^Z=Z\- u \b C(z) \qquad \rarrow \qquad s_\ww u = -v_\epsilon u + u \b c(\epsilon)  .
\end{align*}
The final composite ghost if then clearly an instance of \eqref{NewGhost3}:
\begin{align}
\label{final_ghost2}
v_\ww&=u\- v u + u\- su, \notag \\
   &=u\-(v_\epsilon + v_\ss+v_\iota) u + u\-(s_\ww + s_\l + s_1) u, \notag \\
   &=u\-(v_\epsilon + v_\ss+v_\iota) u + u\- \big( -v_\epsilon u + u \b c(\epsilon) - v_\ss u- v_\iota u    \big) =\b c(\epsilon).
\end{align}
Which proves the assertion since $\b c(\epsilon) =\b k_1(\epsilon)  + \b v_\epsilon=\u\- k_1(\epsilon) \u + v_\epsilon + \t v_\epsilon=\u\-c(\epsilon)\u + \t v_\epsilon$.

 From the explicit form \eqref{final_ghost1} we see first that both conformal boosts ghost $\iota$ and the Lorentz ghost $s$ has disappeared, meaning that $s_1 \chi_\l=0$ and $s_\l\chi_\l=0$. This reflects the fact that the composite fields $\chi_\l$ are $\K_1$ and $\SO$-invariant. Second, from the derivation \eqref{final_ghost2} of the final ghost $v_\l$ in particular, we see that they are Weyl gauge fields of a non-standard kind. They satisfy the algebra $\mathsf{BRST}_\ww$: 
\begin{align*}
s_\ww\varpi_\l&=-D\b c(\epsilon)=-d\b c(\epsilon) -[\varpi_\l, \b c(\epsilon)]=
\begin{pmatrix}[1.2]  0 & - \nabla \d\epsilon
& 0 \\ 0 &\ -d\epsilon \delta - dx \d\epsilon - g\-\!\!\cdot\!\d\epsilon\  g\!\cdot\!dx & \ \
-g\- \!\cdot\!\nabla \d\epsilon - 2\epsilon g\-\!\!\cdot\! \sP \\ 
0 & 2\epsilon g\!\cdot\!dx  &  0\end{pmatrix} 
\end{align*}
where $-\nabla\d\epsilon=-d\d\epsilon - \d\epsilon\!\cdot\!\Gamma=\left( \nabla_\mu (\d_\nu\epsilon) - \d_\alpha\epsilon\ {\Gamma^\alpha}_{\mu\nu} \right)dx^\mu$, 
\begin{align*}
& s_\ww\b\Omega_\l=[\b\Omega_\l, \b c(\epsilon)]=
 \begin{pmatrix}[1.2] -\d\epsilon\!\cdot\!\sT &   - \d\epsilon\!\cdot\!\left( \sW -f_\l \right) & 0 \\
0 &  \sT\d\epsilon -  g\-\!\!\cdot\!\d\epsilon\ g\!\cdot\! \sT\ & \ -2\epsilon g\-\!\!\cdot\!\sC + g\-\!\!\cdot\! (\sW +f_\l)\!\cdot\!\d\epsilon\\
 0 &  2\epsilon g\!\cdot\!\sT  & * \end{pmatrix}, 
 \end{align*}
in the normal case  $\quad s_\ww\b\Omega_{\n, \l}=\begin{pmatrix}[1.2] 0 &   - \d\epsilon\!\cdot\!\sW & 0 \\
0 &  0 & \ -2\epsilon g\-\!\!\cdot\!\sC + g\-\!\!\cdot\!\sW\!\cdot\!\d\epsilon\\
 0 &  0  & 0 \end{pmatrix}.
 $  So this ``pure gauge'' part of $\mathsf{BRST}_\ww$ reproduces the infinitesimal Weyl  transformations of the Linear/Levi-Civita connection and of the metric, Schouten, (generalized) Cotton and Weyl tensors. Compare with \eqref{varpi_L_Z}-\eqref{Omega_L_Z}.
 Finally the part of $\mathsf{BRST}_\ww$ concerning the tractor and ghost fields is
 \begin{align*}
&s_\ww\vphi_\l=-v_\ww\vphi_\l =
\begin{pmatrix} -\epsilon \rho_\l- \d\epsilon\!\cdot\!\ell_\l \\[1mm]  -\epsilon\ell_\l - g\-\!\!\cdot\!\d\epsilon\ \s \\ \epsilon \s \end{pmatrix}=\begin{pmatrix}[1.2] -\epsilon \rho_\l- \d_\nu\epsilon\ \ell_\l^\nu \\[1mm] - \epsilon \ell_\l^\mu - g^{\mu\nu} \d_\nu\epsilon\ \s \\ \epsilon \s \end{pmatrix} 
\qquad \text{and} \qquad s_\ww v_\ww=-{v_\ww}^2=
\begin{pmatrix} 0 & 0  & 0 \\ 0 & 0 & -2\epsilon g\-\!\!\cdot\!\d\epsilon \\ 0 & 0 & 0  \end{pmatrix}.
\end{align*}
The last relation reflects essentially that $s_\ww \epsilon=0$, that is the fact that $\W$ is an abelian group. The first relation reproduces the infinitesimal transformation law of a tractor field, to be compared with  \eqref{Tractor_TConnection_Z}.

\section{Conclusion} 
\label{Conclusion} 

Tractors and twistors are frameworks devised to deal with conformal calculus on manifolds. Whereas it has been noticed that both are vector bundles associated to the conformal Cartan principal bundle endowed with its normal Cartan connection, it is often deemed more direct and intuitive to produce them, bottom-up, from the prolongation of defining differential equations, the Almost Einstein and Twistor equations respectively. In this paper  we have proposed a straightforward and top-down gauge theoretic construction of tractors via the dressing field method of gauge symmetries reduction. 

 The starting point was the conformal Cartan gauge structure $\left ((\P, \varpi), E \right)$ over $\M$ with gauge symmetry given by the gauge group $\H$ - comprising Weyl $\W$, Lorentz $\SO$ and conformal boosts $\K_1$  groups - acting on gauge variables $\chi=\{\varpi, \b\Omega, \vphi, D\vphi\}$ which are the conformal Cartan connection, its curvature, a section of the naturally associated $\doubleR^{n+2}$-vector bundle and its covariant derivative.  

Applying the dressing field approach we showed that  $\K_1$-invariant composite fields $\chi_1$ could be constructed thanks to a dressing field $u_1$ built out of parts of the Cartan connection $\varpi$.
In particular the dressed section $\vphi_1 \in E_1$ was shown to be indeed a tractor and $D_1=d + \varpi_1$  a generalized tractor connection. The usual one is induced by the dressed normal conformal Cartan connection: $D_{\n, 1}=d +\varpi_{\n, 1}$. 
 Furthermore we stressed that, while the composite fields $\chi_1$ are genuine gauge fields w.r.t the residual Lorentz gauge symmetry $\SO$, the latter are gauge fields of a non-standard kind w.r.t the residual Weyl gauge symmetry $\W$. Such non-standard gauge fields resulting from the dressing field method, implement the gauge principle of physics in a satisfactory way but are not of the same geometric nature than the fields usually underlying gauge theories. The trator bundle with connection $(E_1, D_{\n, 1})$, as a restriction of $(E_1, D_1)$, is then seen to be an instance of non-standard gauge structure over $\M$. 

A further dressing, in the form of the vielbein, allowed to further reduce the Lorentz $\SO$-gauge symmetry so as to produce the non-standard composite $\W$-gauge fields $\chi_\l$. These are essentially the same as $\chi_1$ but written in holonomic frames rather than orthonormal frames. This allows to recover the form of the tractor bundles and tractor connection as usually derived, $(\T, \nabla^\T)=(E_\l, D_{\n, \l})$, as a restriction of the non-standard gauge structure $(E_L, D_L)$. 

The initial conformal gauge structure is encoded infinitesimally in the initial $\mathsf{BRST}$ algebra satisfied by the gauge variables $\chi$. As a general result the dressing field approach modifies the BRST algebra of a gauge structure. We then provided the new algebra $\mathsf{BRST}_1$ satisfied by the composite fields $\chi_1$, as well as the algebra $\mathsf{BRST}_\ww$ satisfied by the composited fields $\chi_\l$. 
\medskip

Tractor calculus was originally devised mainly for conformal geometry, but also for projective geometry, see \cite{Bailey-et-al94}. For the latter case also the construction is via prolongation of a differential equation. The gauge theoretic method that we advocate here ought to apply with equal felicity to projective tractors, the latter being  recovered from dressing of the projective Cartan bundle. Standard results concerning the projective Cartan connection \cite{Kobayashi-Nagano1964} should be obtained as well. This will be the subject of a future paper. For more immediate concern, in a  companion paper we will deal with the application of the dressing approach to twistors, which can be seen as derived from the spinor bundle associated to the conformal Cartan bundle $(\P, \varpi)$.

\section*{Acknowledgement}  

We wish to thank Serge \textsc{Lazzarini} and Thierry \textsc{Masson} (CPT Marseille) for their encouragements and  for supporting discussions while this work was under completion. A special thank is due to Thierry Masson for bringing to our attention the notion of $1$-$\alpha$-cocycle.

{
\small
 \bibliography{Biblio}

\begin{thebibliography}{10}

\bibitem{Weyl1918}
H.~Weyl.
\newblock {Gravitation and electricity}.
\newblock {\em Sitzungsber. Preuss. Akad. Wiss. Berlin (Math. Phys.)},
  1918:465, 1918.

\bibitem{Weyl1919}
H.~Weyl.
\newblock {A New Extension of Relativity Theory}.
\newblock {\em Annalen Phys.}, 59:101--133, 1919.
\newblock [Annalen Phys.364,101(1919)].

\bibitem{ORaif}
L.~O'Raifeartaigh.
\newblock {\em The Dawning of Gauge Theory}.
\newblock Princeton Series in Physics. Princeton University Press, 1997.

\bibitem{Bach1921}
R.~Bach.
\newblock {Zur Weylschen Relativit{\"a}tstheorie und der Weylschen Erweiterung
  des Kr{\"u}mmungstensorbegriffs}.
\newblock {\em Mathematische Zeitschrift}, 9(1-2):110--135, 1921.

\bibitem{Utiyama1956}
R.~Utiyama.
\newblock {Invariant Theoretical Interpretation of Interaction}.
\newblock {\em Phys. Rev.}, 101:1597--1607, Mar 1956.

\bibitem{Cham-West1977}
A.H. Chamseddine and P.C. West.
\newblock {Supergravity as a Gauge theory of supersymmetry}.
\newblock {\em Nucl. Phys. B}, 129:39--44, 1977.

\bibitem{Townsend1977}
P.~K. Townsend.
\newblock {Cosmological constant in Supergravity}.
\newblock {\em Phys. Rev. D}, 15:2802--2804, 1977.

\bibitem{McDowell-Mansouri}
S.~W. McDowell and F.~Mansouri.
\newblock Unified geometric theory of gravity and supergravity.
\newblock {\em Physical Review Letters}, 38:739--742, 1977.

\bibitem{West1978}
P~C. West.
\newblock {A geometric gravity lagrangian}.
\newblock {\em Phys. Lett.}, 76B:569--570, 1978.

\bibitem{Stelle-West1979}
K~S. Stelle and P~C. West.
\newblock {de Sitter gauge invariance and the geometry of the Einstein-Cartan
  theory}.
\newblock {\em J. Phys. A}, 12:205--210, 1979.

\bibitem{Fulton_et_al1962}
T.~Fulton, F.~Rohrlich, and L.~Witten.
\newblock Conformal invariance in physics.
\newblock {\em Rev. Mod. Phys.}, 34:442--457, Jul 1962.

\bibitem{Cunningham}
Cunningham.
\newblock The principle of relativity in electrodynamics and an extension
  thereof.
\newblock {\em Proc. London Math. Soc.}, 8:77--98, 1910.

\bibitem{Bateman1909}
H.~Bateman.
\newblock The transformation of the electrodynamical equations.
\newblock {\em Proc. London Math. Soc.}, 8:223--264, 1910.

\bibitem{Bateman1910}
H~Bateman.
\newblock The transformations of coordinates which can be used to transform one
  physical problem into another.
\newblock {\em Proc. London. Math. Soc.}, 8:469--488, 1910.

\bibitem{CrispimRomao1977}
J.~{Crispim Romao}, A.Ferber, and P.G.O. Freund.
\newblock {Unified Superconformal Gauge Theories}.
\newblock {\em Nucl. Phys.}, B126:429, 1977.

\bibitem{Ferber-Freund1977}
A.~Ferber and P.G.O. Freund.
\newblock {Superconformal Supergravity and Internal Symmetry}.
\newblock {\em Nucl. Phys.}, B122:170, 1977.

\bibitem{Kaku_et_al1977}
M.~Kaku, P.~K. Townsend, and P.~Van Nieuwenhuizen.
\newblock {Gauge theory of the conformal and superconformal group}.
\newblock {\em Phys. Lett.}, 69B:304--308, 1977.

\bibitem{Kaku-et-al1978}
M.~Kaku, P.~K. Townsend, and P.~van Nieuwenhuizen.
\newblock {Properties of Conformal Supergravity}.
\newblock {\em Phys. Rev.}, D17:3179, 1978.

\bibitem{Harnad-Pettitt1}
J.P. Harnad and R.B. Pettitt.
\newblock {Gauge theories for space-time symmetries}.
\newblock {\em J. Math. Phys.}, 17:1827--1837, 1976.

\bibitem{Harnad-Pettitt2}
J.P. Harnad and R.B Pettitt.
\newblock {Gauge Theory of the Conformal Group}.
\newblock In T.~Sharp and B.~Kolman, editors, {\em {Group Theoretical Methods
  in Physics (Proceedings of the Fifth International Colloquium, Montr{\'e}al
  1976)}}, pages 277--301. Academic Press, Inc., 1977.

\bibitem{Penrose1960}
R.~Penrose.
\newblock {A Spinor approach to general relativity}.
\newblock {\em Annals Phys.}, 10:171--201, 1960.

\bibitem{Penrose1967}
R.~Penrose.
\newblock {Twistor algebra}.
\newblock {\em J. Math. Phys.}, 8:345, 1967.

\bibitem{Penrose1968}
Roger Penrose.
\newblock {Twistor quantization and curved space-time}.
\newblock {\em Int. J. Theor. Phys.}, 1:61--99, 1968.

\bibitem{Penrose-McCallum72}
R.~Penrose and M.A.H. MacCallum.
\newblock Twistor theory: An approach to the quantisation of fields and
  space-time.
\newblock {\em Physics Reports}, 6(4):241 -- 316, 1973.

\bibitem{Goenner2004}
Hubert F.~M. Goenner.
\newblock On the history of unified field theories.
\newblock {\em Living Reviews in Relativity}, 7(2), 2004.

\bibitem{Wise09}
D.~K. Wise.
\newblock Symmetric space, cartan connections and gravity in three and four
  dimensions.
\newblock {\em SIGMA}, 5:080--098, 2009.

\bibitem{Wise10}
D.~K. Wise.
\newblock {MacDowell-Mansouri gravity and Cartan geometry}.
\newblock {\em Classical and Quantum Gravity}, 27:155010, 2010.

\bibitem{Bailey-et-al94}
T.N. Bailey, M.G. Eastwood, and A.R. Gover.
\newblock Thomas's structure bundle for conformal, projective and related
  structures.
\newblock {\em Rocky Mountain J. Math.}, 24(4):1191--1217, 12 1994.

\bibitem{Gover-Shaukat-Waldron09}
Gover.~A. R., Shaukat. A, and Waldron. A.
\newblock Tractors, mass and weyl invariance.
\newblock {\em Nuclear Physics B}, 812(3):424--455, May 2009.

\bibitem{Gover-Shaukat-Waldron09-2}
Gover.~A. R., Shaukat. A, and Waldron. A.
\newblock Weyl invariance and the origins of mass.
\newblock {\em Physics Letters B}, 675(1):93--97, May 2009.

\bibitem{Friedrich77}
Helmut Friedrich.
\newblock Twistor connection and normal conformal cartan connection.
\newblock {\em General Relativity and Gravitation}, 8(5):303--312, 1977.

\bibitem{Curry-Gover2015}
Sean Curry and A.~Rod Gover.
\newblock {An introduction to conformal geometry and tractor calculus, with a
  view to applications in general relativity}.
\newblock {\em arXiv:1412.7559}, 2014.

\bibitem{Penrose-Rindler-vol2}
Roger Penrose and Wolfgang Rindler.
\newblock {\em Spinors and Space-Time}, volume~2.
\newblock Cambridge University Press, 1986.

\bibitem{Bailey-Eastwood91}
M~Eastwood and Toby Bailey.
\newblock Complex paraconformal manifolds - their differential geometry and
  twistor theory.
\newblock {\em Forum mathematicum}, 3(1):61--103, 1991.

\bibitem{Francois2014}
J.~Fran\c{c}ois.
\newblock {\em {Reduction of gauge symmetries: a new geometrical approach}}.
\newblock Theses, {Aix-Marseille Universit{\'e}}, September 2014.

\bibitem{Dubois-Violette1987}
M.~Dubois-Violette.
\newblock The {W}eil-{BRS} algebra of a {L}ie algebra and the anomalous terms
  in gauge theory.
\newblock {\em J. Geom. Phys}, 3:525--565, 1987.

\bibitem{Stora1984}
R.~Stora.
\newblock {Algebraic structure and toplogical origin of chiral anomalies}.
\newblock In G.~{'t Hooft} and et~al., editors, {\em {Progress in Gauge Field
  Theory, Cargese 1983}}, {NATO ASI Ser.B, Vol.115}. Plenum Press, 1984.

\bibitem{Manes-Stora-Zumino1985}
J.~Ma{\~n}es, R.~Stora, and B.~Zumino.
\newblock Algebraic study of chiral anomalies.
\newblock {\em Commun. Math. Phys.}, 102:157--174, 1985.

\bibitem{Baulieu-TMieg1984}
L.~Baulieu and J.~Thierry-Mieg.
\newblock {Algebraic Structure of Quantum Gravity and the Classification of the
  Gravitational Anomalies}.
\newblock {\em Phys. Lett.}, B145:53, 1984.

\bibitem{Baulieu-Bellon1986}
L.~Baulieu and M.~Bellon.
\newblock {$p$-Forms and Supergravity: Gauge Symmetries in Curved Space}.
\newblock {\em Nucl. Phys.}, B266:75, 1986.

\bibitem{Bonora-Cotta-Ramusino}
L.~Bonora and P.~Cotta-Ramusino.
\newblock Some remark on brs transformations, anomalies and the cohomology of
  the lie algebra of the group of gauge transformations.
\newblock {\em Commun. Math. Phys.}, 87:589--603, 1983.

\bibitem{Masson-Wallet}
T.~Masson and J.~C. Wallet.
\newblock A remark on the spontaneous symmetry breaking mechanism in the
  standard model.
\newblock {\em arXiv:1001.1176}, 2011.

\bibitem{GaugeInvCompFields}
C.~Fournel, J.~Fran{\c c}ois, S.~Lazzarini, and T.~Masson.
\newblock Gauge invariant composite fields out of connections, with examples.
\newblock {\em Int. J. Geom. Methods Mod. Phys.}, 11(1):1450016, 2014.

\bibitem{Pedersen79}
G.~K. Pedersen.
\newblock {\em C-star Algebras and Their Automorphisms Groups}.
\newblock London Mathematical Society Monographs. Academic Press, Inc., 1979.

\bibitem{Williams07}
Dana.~P. Williams.
\newblock {\em Crossed Products of C-star Algebras}, volume 134 of {\em
  Mathematical Surveys and Monographs}.
\newblock American Mathematical Society, 2007.

\bibitem{Dirac55}
P.~A.~M. Dirac.
\newblock Gauge-invariant formulation of quantum electrodynamics.
\newblock {\em Canadian Journal of Physics}, 33:650--660, 1955.

\bibitem{Dirac58}
P.~A.~M. Dirac.
\newblock {\em The principles of Quantum Mechanics}.
\newblock Oxford University Press, 4th edn edition, 1958.

\bibitem{Higgs66}
Peter~W. Higgs.
\newblock Spontaneous symmetry breakdown without massless bosons.
\newblock {\em Phys. Rev.}, 145:1156--1163, May 1966.

\bibitem{Kibble67}
T.~W.~B. Kibble.
\newblock {Symmetry breaking in nonAbelian gauge theories}.
\newblock {\em Phys. Rev.}, 155:1554--1561, 1967.

\bibitem{Pervushin}
V.~Pervushin.
\newblock Dirac variables in gauge theories.
\newblock {\em arXiv:hep-th/0109218v2}, 2001.

\bibitem{Lantsman}
L.~D. Lantsman.
\newblock Dirac fundamental quantization of gauge theories is the natural way
  of reference frames in modern physics.
\newblock {\em Fizika B}, 18:99--140, 2009.

\bibitem{Lavelle-McMullan93}
Martin Lavelle and David McMullan.
\newblock Nonlocal symmetry for qed.
\newblock {\em Phys. Rev. Lett.}, 71:3758--3761, Dec 1993.

\bibitem{McMullan-Lavelle97}
M.~Lavelle and D.~McMullan.
\newblock Constituent quarks from {QCD}.
\newblock {\em Physics Reports}, 279:1--65, 1997.

\bibitem{LorceGeomApproach}
C.~Lorc\'e.
\newblock Geometrical approach to the proton spin decomposition.
\newblock {\em Physical Review D}, 87:034031, 2013.

\bibitem{Leader-Lorce}
E.~Leader and C.~Lorc\'e.
\newblock The angular momentum controversy: What is it all about and does it
  matter?
\newblock {\em Physics Reports}, 514:163--248, 2014.

\bibitem{FLM2015_I}
J.~Fran\c{c}ois, S.~Lazzarini, and T.~Masson.
\newblock Nucleon spin decomposition and differential geometry.
\newblock {\em Phys. Rev. D}, 91:045014, Feb 2015.

\bibitem{Frohlich-Morchio-Strocchi81}
J.~Frohlich, G.~Morchio, and F.~Strocchi.
\newblock Higgs phenomenon without symmetry breaking order parameter.
\newblock {\em Nuclear Physics B}, 190(3):553 -- 582, 1981.

\bibitem{McMullan-Lavelle95}
Martin Lavelle and David McMullan.
\newblock Observables and gauge fixing in spontaneously broken gauge theories.
\newblock {\em Physics Letters B}, 347(1):89 -- 94, 1995.

\bibitem{Chernodub2008}
M.~N. Chernodub, Ludvig Faddeev, and Antti~J. Niemi.
\newblock {Non-abelian Supercurrents and Electroweak Theory}.
\newblock {\em JHEP}, 12:014, 2008.

\bibitem{Faddeev2009}
L.~D. Faddeev.
\newblock {\em An Alternative Interpretation of the Weinberg-Salam Model},
  pages 3--8.
\newblock Springer Netherlands, Dordrecht, 2009.

\bibitem{Ilderton-Lavelle-McMullan2010}
Anton Ilderton, Martin Lavelle, and David McMullan.
\newblock Symmetry breaking, conformal geometry and gauge invariance.
\newblock {\em Journal of Physics A: Mathematical and Theoretical},
  43(31):312002, 2010.

\bibitem{Struyve2011}
Ward Struyve.
\newblock Gauge invariant accounts of the higgs mechanism.
\newblock {\em Studies in History and Philosophy of Science Part B: Studies in
  History and Philosophy of Modern Physics}, 42(4):226 -- 236, 2011.

\bibitem{vanDam2011}
Suzanne van Dam.
\newblock Spontaneous symmetry breaking in the higgs mechanism.
\newblock {\em PhiSci-Archive}, 2011.

\bibitem{Garajeu-Grimm-Lazzarini}
D.~Garajeu, R.~Grimm, and S.~Lazzarini.
\newblock W-gauge structures and their anomalies: An algebraic approach.
\newblock {\em Journal of Mathematical Physics}, 36:7043--7072, 1995.

\bibitem{Polyakov1989}
Alexander~M. Polyakov.
\newblock {Gauge Transformations and Diffeomorphisms}.
\newblock {\em Int. J. Mod. Phys.}, A5:833, 1990.

\bibitem{Lazzarini2008}
Serge Lazzarini and Carina Tidei.
\newblock Polyakov soldering and second-order frames: The role of the cartan
  connection.
\newblock {\em Letters in Mathematical Physics}, 85(1):27--37, 2008.

\bibitem{Attard-Lazz2016}
J.~Attard and S.~Lazzarini.
\newblock A note on weyl invariance in gravity and the wess-zumino functional.
\newblock {\em Nuclear Physics B}, 2016.

\bibitem{Cap-Slovak09}
Andreas Cap and Jan Slovak.
\newblock {\em Parabolic Geometries I: Background and General Theory}, volume~1
  of {\em Mathematical Surveys and Monographs}.
\newblock American Mathematical Society, 2009.

\bibitem{Sharpe}
R.~W. Sharpe.
\newblock {\em Differential Geometry: Cartan's Generalization of Klein's
  Erlangen Program}, volume 166 of {\em Graduate text in Mathematics}.
\newblock Springer, 1996.

\bibitem{Kobayashi}
S.~Kobayashi.
\newblock {\em Transformation Groups in Differential Geometry}.
\newblock Springer, 1972.

\bibitem{FLM2015_II}
J.~Fran\c{c}ois, S.~Lazzarini, and T.~Masson.
\newblock {Residual Weyl symmetry out of conformal geometry and its BRST
  structure}.
\newblock {\em JHEP}, 09:195, 2015.

\bibitem{Kobayashi-Nagano1964}
S.~Kobayashi and T.~Nagano.
\newblock On projective connections.
\newblock {\em Journal of Mathematics and Mechanics}, 13:215--235, 1964.

\end{thebibliography}
}

\end{document}